\documentclass[envcountsame]{llncs}

\usepackage{amssymb,amsmath}
\usepackage{epsfig}

\usepackage{graphicx}
\usepackage{color}

\usepackage{ifthen}
\newcounter{doctype}
\newcounter{prooftype}

\newcommand{{\Nat}}{{\rm I\! N}}
\newcommand{{\Real}}{{\rm I\! R}}

\newcommand{{\prob}}{{\rm \bf P}}
\newcommand{{\mean}}{{\rm \bf E}}
\newcommand{{\var}}{{\rm \bf V}}

\newcommand{{\myalp}}[0]{\ensuremath{\mathcal{A}}}
\newcommand{{\myalple}}[1]{\ensuremath{\mathcal{A}^{#1}}}
\newcommand{{\myalpinf}}[0]{\ensuremath{\mathcal{A}^\infty}}
\newcommand{{\myalpfin}}[0]{\ensuremath{\mathcal{A}^{<\infty}}}
\newcommand{{\myalpall}}[0]{\ensuremath{\mathcal{A}^{\le\infty}}}

\newcommand{{\myalpb}}[0]{\ensuremath{\mathcal{B}}}
\newcommand{{\myalpble}}[1]{\ensuremath{\mathcal{B}^{#1}}}
\newcommand{{\myalpbinf}}[0]{\ensuremath{\mathcal{B}^\infty}}
\newcommand{{\myalpbfin}}[0]{\ensuremath{\mathcal{B}^{<\infty}}}
\newcommand{{\myalpball}}[0]{\ensuremath{\mathcal{B}^{\le\infty}}}

\newcommand{{\wret}}[2]{\ensuremath{A^{(#1)}(#2)}}
\newcommand{{\wreton}}[2]{\ensuremath{B^{(#1)}(#2)}}
\newcommand{{\toalp}}[4]{\ensuremath{\mu\!\left(#1,#2[#3\!\ldots\!#4]\right)}}
\newcommand{{\toalpx}}[5]{\ensuremath{\mu_{#1}\!\left(#2,#3[#4\!\ldots\!#5]\right)}}
\newcommand{{\toalpxy}}[3]{\ensuremath{\mu_{#1}\!\left(#2,#3\right)}}
\newcommand{{\toalpret}}[4]{\ensuremath{\tilde{\mu}\!\left(#1,#2[#3\!\ldots\!#4]\right)}}
\newcommand{{\toalpretx}}[5]{\ensuremath{\tilde{\mu}_{#1}\!\left(#2,#3[#4\!\ldots\!#5]\right)}}
\newcommand{{\myread}}[2]{\ensuremath{\rho_{#1,#2}}}
\newcommand{{\mydel}}[1]{\ensuremath{\rho_{#1}}}

\newcommand{{\df}}{{\mathrm{def}}}
\newcommand{{\argmax}}{{\mathrm{arg max}}}
\newcommand{{\rank}}{{\mathrm{rank}}}
\newcommand{{\proj}}{{\mathrm{proj}}}
\newcommand{{\state}}{{\mathrm{STATE}}}
\newcommand{{\mess}}{{\mathrm{MESS}}}
\newcommand{{\lcp}}{{\mathrm{lcp}}}
\newcommand{{\bigo}}{{\mathcal{O}}}

\newcommand{{\SUB}}{{\mathrm{SUB}}}
\newcommand{{\INS}}{{\mathrm{INS}}}
\newcommand{{\DEL}}{{\mathrm{DEL}}}
\newcommand{{\ECHO}}{{\mathrm{ECHO}}}

\newcommand{\myprob}[1]{\ensuremath{\prob\!\left\{#1\right\}}}
\newcommand{\mymean}[1]{\ensuremath{\mean\!\left[#1\right]}}
\newcommand{{\mycom}}[2]
{{{\em   \textcolor{red}{\textsf{#1}} } :
    \textsf{#2}
}}

\begin{document}

\mainmatter  


\title{Smoothed Analysis of Trie Height \\ by Star-like PFAs}


%
%
\author{Stefan Eckhardt\inst{1} 
\and Sven Kosub\inst{2}
\and Johannes Nowak\inst{1}
}
%

\institute{Fakult\"at f\"ur Informatik, Technische Universit\"at M\"unchen,\\
Boltzmannstra{\ss}e 3, D-85748 Garching, Germany\\
\texttt{\{eckhardt,nowakj\}@in.tum.de}
\and
Fachbereich Informatik und Informationswissenschaft, Universit\"at Konstanz,\\
Box D 67, D-78457 Konstanz, Germany\\
\texttt{kosub@inf.uni-konstanz.de}
}

%
%

\maketitle

\begin{abstract}
Tries are general purpose data structures for information retrieval. 
The most significant parameter of a trie is its height $H$ which equals the 
length of the longest common prefix of any two string in the set $A$ over 
which the trie is built. 
Analytical investigations of random tries suggest that 
$\mymean{H}\in O(\log(\|A\|))$, although $H$ is unbounded  in the worst case. 
Moreover, sharp results on the distribution function of $H$ are known 
for many different random string sources. 
But because of the inherent weakness of the modeling behind average-case 
analysis---analyses being dominated by random data---these results 
can utterly explain the fact that in many practical situations the trie
height is logarithmic. 
We propose a new semi-random string model and perform a smoothed analysis 
in order to give a mathematically more rigorous explanation for the 
practical findings.
The perturbation functions which we consider are based on probabilistic 
finite automata (PFA) and we show that the transition probabilities of the
representing PFA completely characterize the asymptotic growth of the smoothed 
trie height. 
Our main result is of dichotomous nature--logarithmic or unbounded---and is 
certainly not surprising at first glance, but we also give quantitative upper 
and lower bounds, which are derived using multivariate generating function in 
order to express the computations of the perturbing PFA. 
A direct consequence is the logarithmic trie height for edit perturbations
(i.e., random insertions, deletions and substitutions). 

\end{abstract}

\section{Introduction}

\paragraph{Motivation.} 
Tries are very simple general purpose data structures for information retrieval. 
This explains why many parameters of tries, such as height, path length or size 
have been and are still subject to extensive average-case analysis under various 
random string models. 
Though almost all investigations of trie height using analytical methods suggest 
the height of a random trie to be logarithmic in the number of strings, it is not 
immediately clear that these results can utterly explain the fact that in many 
practical settings the height is in fact logarithmic in the number of strings and 
thus far from its worst case.
This holds particularly in the case of non-random data.  
Nilsson and Tikkanen~\cite{Nilsson-2002} have experimentally investigated the 
height of PATRICIA trees, or path-compressed tries, and other search structures.  
There, the height of a PATRICIA tree, built over a set of 50,000 unique random 
uniform strings was 16 on average and 20 at most.  
For non-random data consisting of 19,461 strings from geometric data, of 16,542 
ASCII character strings from a book, and of 38,367 strings from Internet routing 
tables, the height of a path-compressed trie, built over these data sets, was on
average 21, 20, and 18, respectively, and at most 30, 41 and 24, respectively.  
These findings suggest that worst-case inputs, i.e., sets for which the height 
of the respective trie is unbounded, are isolated peaks in the input space and
even small deviations from worst-case inputs yield logarithmic trie height.  
In this work we try to give an analytical explanation of these findings.  

The previous average-case approaches typically suffer from two drawbacks: such 
analyses are usually dominated by a high proportion of purely random inputs, 
even if the random inputs are produced by very sophisticated random string models
such as the recently introduced symbolic dynamical systems~\cite{Valle-2001}; 
moreover, even those results that give sharper bound on the higher moments of 
the distribution function of $H_S$ cannot explain the behavior of a trie on an
input that is very close to worst-case.  
{\em Smoothed analysis},
introduced by Spielman and Teng in their seminal paper~\cite{Spielman-2004} in 
order to explain the good practical performance of the simplex algorithm which 
is opposed to its bad worst-case behavior, gives a mathematical framework to 
better understand such findings: one is not interested in finding a probability
distribution which models the typical input space more accurately.  
Rather, one aims at answering the following kind of question: are worst-case 
inputs ``isolated peaks'' or ``plateaus''? 
To this end, the smoothed complexity of an algorithm---or more generally of a 
random variable---is defined as the maximum over all inputs of the expected 
running time of the algorithm under slight random perturbations of the respective 
input.  
In order to perform a meaningful smoothed analysis, one must find an {\em adequate} 
perturbation function, i.e., one which resembles those random influences which
real world inputs are typically subject to.
 
In order to perform a meaningful smoothed analysis of the most significant parameter 
of a trie, namely its height, we present a new {\em semi-random} model for strings: 
the set of input strings is chosen in advance by an adversary and then strings are 
randomly perturbed independently using the same perturbation functions. 
The adversary has full information on the parameters of the perturbation function, 
but has no control over the random perturbations and the parameters, once the input 
set is chosen. 
This model fits into the framework of smoothed analysis. 
(A somewhat stronger model for semi-random sources was considered by Santha and 
Vazirani in~\cite{Santha-1986}, though it was not in the context of tries but in 
the context of random and quasi-random number generators: there, the adversary had 
(limited) control over each of the biases in a sequence biased coin flips and 
full knowledge over the previous history.)  
The class of string perturbation functions which we consider can be represented 
by (Mealy-type) probabilistic finite automata (PFAs).
PFAs are a standard tool for modeling unreliable deterministic systems and  
they provide a compact representation for a very natural class of string 
perturbation functions, namely {\em random edit perturbations}, which occur in 
those settings and thus resembles some of the typical random influences that 
strings are exposed to.  
To the best of our knowledge, we are the first to perform a smoothed analysis 
of trie parameters.
 
\paragraph{Results.} 

The main technical contribution of this paper is a characterization of the 
smoothed trie height depending on the probabilistic automaton underlying the 
perturbation function. 
For a star-like perturbation automaton, it is logarithmic if and only if 
certain conditions for the automaton's transitions hold; if the conditions do 
not hold then the height is unbounded (see Theorem \ref{thm:star-like}).  
The logarithmic/unbounded-height dichotomy is certainly not surprising,
but the conditions are very easy to check.  
So, the theorem can be applied to rather complex perturbation models for
which an ad-hoc analysis appears quite involved. 
In order to derive the result, it turns out that we must bound the
coincidence probability of length $k$ by an exponentially decreasing term in 
$k$. 
To do so, we use multivariate rational generating functions to express the 
computations of the perturbing PFA.  
This approach, which is called the \emph{weighted words model} (cf.~\cite{Flajolet-2007}),
seems to fit best the requirements of our analysis.  
A direct consequence of the theorem is a proof of the logarithmic smoothed 
trie height for random edit perturbations (i.e., insertions, deletions,
substitutions).  
We should note that not all plausible string perturbation functions can be 
modeled by star-like automata, e.g., transpositions.

\medskip

\noindent
{\em Due to the page limit, all technical proofs of this paper are omitted.
Instead, they can be found in the full paper \cite{Eckhardt-2007} 
(or in the appendix).}

\section{Preliminaries}

Let $\Nat=\{0,1,2,\dots\}$ and $\Nat_+=\{1,2,\dots\}$.
 Let $\myalp$ denote the finite alphabet. 
The elements of $\myalp$ are called the {\em symbols} of the alphabet.  
For $m\in\Nat$, the finite {\em sequence} $s=(a_1,\ldots,a_m)$ of symbols 
$a_i\in\myalp$ is called a {\em finite string} over $\myalp$ of {\em length} $m$, 
denoted by $|s|$.  
If $m=0$ then the string is called the {\em empty string} and is denoted by 
$\epsilon$.  
An infinite sequence $s=(a_1,a_2,\ldots)$ of symbols such that for $i\in\Nat_+$ 
it holds that $a_i\in\myalp$ is called and {\em infinite string}. 
In this case, we set $|s|=+\infty$.  
A string $s=(a)$ of length one will by abbreviated by $a$.  
For a finite string $s=(a_1,\ldots,a_m)$ of length $m$ and $i\in\{1,\ldots,m\}$ 
we access the $i$-th element $a_i$ by $s[i]$. 
Also, for an infinite string $s$ we access the $i$-th element for $i\in\Nat$ by
$s[i]$ and for every string $s$ it holds the $s[0]=\epsilon$.  
For a finite string $s$ and $i,j\in\{1,\ldots,|s|\}$ satisfying $i<j$, the 
subsequence $(a_i,\ldots,a_{j})$ is called a substring of $s$ and is accessed 
by $s[i\ldots j]$. 
Here, for $i,j\in\{1,\ldots,m\}$ satisfying $i>j$ we define $s[i\ldots j]=\epsilon$ 
as the access to the empty string.
If $s$ is infinite, we access the infinite substring starting at the $i$-th 
position of $s$ by $s[i\ldots]$.
For a symbol $a\in\myalp$ and a string $s$ over the same alphabet we denote by 
$|s|_a$ the number of occurrences of the symbol $a$ in $s$. 
For a finite string $s$, it clearly holds that  $|s|=\sum_{a\in\myalp}|s|_a$.
For  a natural number $m$ we denote by $\myalple{m} $ the set of all strings over $\myalp$
that have length {\em  exactly $m$} and by $\myalp^{\le m}$ the set of all strings that 
have length {\em  at most $m$}.
Let $\myalpinf $ denote the set of all infinite strings over $\myalp$, 
let $\myalpfin $ denote the set of all finite strings over $\myalp$, and let $\myalpall $ 
denote the set of all finite and infinite strings over $\myalp$. 
A  string $s$ is a {\em  prefix} of a string $t$, if $|s|\le|t|$ and for all 
indices $i\in\{1,\ldots,|s|\}$ it holds that $s[i]=t[i]$. 
We write $s\sqsubseteq t$ in this case. 
A prefix $s$ of $t$ is a {\em  proper prefix}, if $|s|<|t|$. 
We write $s\sqsubset t$ in this case. 
Note that for the empty string $\epsilon\sqsubset t$ for every non-empty string $t$.

\section{Towards Smoothed Trie Height}
\label{sec:towards-smooth-trie}

 \subsection{Related Studies: The Height of Random Tries}
  \label{sec:related-work} 

Let $\myalp=\{a_1,\ldots,a_r\}$ be a finite alphabet of cardinality $r\ge 2$ and 
let $A\subseteq \myalpinf$ be a set of $\|A\|=n$ distinct strings. 
Tries were first introduced and analyzed by Fredkin~\cite{Fredkin-1960} and
 Knuth~\cite{Knuth-1997}.
 For the analysis of random tries, i.e., tries built over a set of random strings, 
the $n$-dimensional product space $\Omega=\myalpinf\times\dots\times\myalpinf$ 
together with some joint probability function $\mu:\Omega\rightarrow[0,1]$ 
constitutes the probability space.  
For an $r$-ary trie built over the set $A$ it holds that the height of the trie 
$H_A=\max_{s,t\in A} \lcp(s,t)$, where $\lcp:\myalpall\times\myalpall\rightarrow\Nat_+$ 
measures the length of the longest common prefix of two strings.  
To analyze its behavior, $H_A$ is viewed as a random variable over the above sample 
space $\Omega$.  
Clearly, in the worst case $H_A$ is unbounded for standard tries.  
By choosing some joint probability function, one can analyze the expected value of 
$H_A$ and other asymptotic properties, e.g., its asymptotic density. 
This has been done for various kinds of probability density functions, where in 
general the $n$ strings in the set $A$ are assumed to be independent and 
identically distributed.  
Thus, the joint density function is completely characterized by the density function $\tilde{\mu}:\myalpinf\rightarrow[0,1]$ for one random string. 
Let $Z$ be a random variable that takes values from $\myalp$. 
Then the one-sided infinite sequence $\{Z_i\}_{i=1}^{\infty}$ can be considered
a random string over $\myalp$.

 The oldest model is the {\em memory-less random source}, were each 
symbol corresponds to a possible outcome of a Bernoulli trial~\cite{Knuth-1997}.  
This means, we are given a parameter vector $p=(p_1,\ldots,p_r)\in(0,1)^r$ and
 for all $i\in\Nat_+$ it holds that $\myprob{Z_i=a_j}=p_j$.  
Another model for random strings that is discussed intensively in the literature 
are {\em Markovian sources}~\cite{Szpankowski-1991,Apostolico-1992}: a string can 
be considered the outcome of transitions of a finite and ergodic Markov chain 
with state space $\myalp$ which has reached its stationary distribution.  
These two models can be subsumed under a the wider class of random strings which 
satisfy the mixing property.  
Pittel~\cite{Pittel-1985,Pittel-1986} considered the growth of different types of 
random trees under the assumption that the underlying random process 
$\{Z_i\}_{i=1}^{\infty}$ satisfies the {\em mixing property}: the sequence 
$\{Z_i\}_{i\ge 1}$ satisfies the mixing property, if there exists $n_0\in\Nat$ and
positive constants $c_1,c_2$ such that for {\em all} $1\le m\le m+n_0\le n$ and 
$A\in\mathfrak{F}^{m}_{1}$ and $B\in\mathfrak{F}^{n}_{m+n_0}$ it holds that
$c_1\cdot\myprob{A}\myprob{B}\le \myprob{A\cap B}\le c_2\cdot\myprob{A}\myprob{B}$\mbox{,}
where for $1\le k\le l$, $\mathfrak{F}^{l}_{k}$ denotes the $\sigma$-field generated 
by the subsequence $\{Z_i\}^{l}_{i=k}$. 
Under this assumption the following limit---the R{\'e}nyi entropy of second order---exists\footnote{originally referred to
as $h_3$ in~\cite{Pittel-1985,Pittel-1986}, but we drop the subscript}
\begin{equation}  \label{eq:RenyiEntropy}
h=\lim_{n\rightarrow\infty} \frac{-\ln{\sum_{\alpha\in\myalple{n}}\myprob{Z_{1}^{n}=\alpha}^2}}{2n}\mbox{,}
\end{equation}
where $Z_{1}^{n} = (Z_1,\ldots,Z_n)$, and the height $H_\textrm{MM}(n)$ of a random 
trie built over a set of $n$ independent strings produced by a mixing source 
satisfies
\begin{equation}
  \label{eq:Height:Mixing}
 H_\textrm{MM}(n)\stackrel{\textrm{w.h.p.}}{\rightarrow}(\ln{n})/h\mbox{.}
\end{equation}
Devroye~\cite{Devroye-1982,Devroye-1984,Devroye-1992a} has introduced the 
{\em density model}, where each string can be considered the fractional binary 
expansion of a random variable from $[0,1)$ and all $n$ random variables are 
assumed to be independent having identical density. 
Particularly, it was shown that the height $H_\textrm{DM}(n)$ of a random trie 
under the density model satisfies 
\begin{equation}
  \label{eq:Density}
-1\le
\liminf_{n\rightarrow \infty} \mymean{H_\textrm{DM}(n)} - \frac{\ln{\alpha}+e}{\ln{2}}
\le
\limsup_{n\rightarrow \infty} \mymean{H_\textrm{DM}(n)} - \frac{\ln{\alpha}+e}{\ln{2}}
\le 1
  \mbox{,}
\end{equation}
if $\int f^2(x)\!\,d x<\infty$, and is unbounded, otherwise. 
Here, $\alpha = \frac{n^2\int_{0}^{1}\!f^2(x)\,dx}{2}$ and $e=2.718\ldots$ is 
Euler's constant. 
Note that this model for random strings accounts for unlimited dependency between
symbols.
Another model, that allows for unlimited dependency are {\em  symbolic 
dynamical systems} which were introduced by Valle{\'{e}}~\cite{Valle-2001} as a
very general model for random strings.
Cl{\'{e}}ment, Valle{\'{e}} and Flajolet~\cite{Clement-2001} have analyzed the 
height of random tries under this model for random strings.

 \subsection{Smoothed Trie Height}
 \label{sec:smoothed-trie-height}

Depending on the real world application in which the tries are used, the previous 
analyses of trie height and other trie parameters give satisfactionary explanations 
of their good practical performance, which is opposed to their bad worst-case 
behavior: if successive data items are independent then the analyses with respect 
to the memory-less random source provide a sound mathematical explanation for the 
practical findings. 
If, on the other hand, data items are not independent, then there are many 
situations in which the analyses with respect to the Markovian source give adequate
answers. 
Nevertheless, none of the results on the height of random tries can be accounted 
for a thorough explanation of the practical findings: this is particularly the 
case in situations where tries are built over natural languages or biological data 
like DNA or protein sequences. 
Those analyses which use random string models suffer from the following two drawbacks 
of average-case analyses: first, it is unclear to which amount the analyses are 
dominated by purely random inputs; second, even the w.h.p. results and relatively 
exact knowledge of the distribution function of the height cannot explain the 
behavior of tries on nearly-worst-case inputs.
To answer these kind of questions, it seems appropriate to perform a smoothed 
analysis and to model a string by means of a {\em  semi-random model}, where
non-random inputs are subject to slight random perturbations. 
We initiate this line of research by performing a smoothed analysis of the most 
crucial parameter of a trie, i.e., its height.
Having motivated the need of a smoothed analysis of trie parameters, we now turn
to the formal definition of the smoothed trie height $H(S,n,X)$. 
Here, $S$ and $X$ denote the input set and the string perturbation function,
respectively, and $n$ is the number of strings that are stored in the trie. 
\begin{definition}
  \label{def:smoothed-trie-height}
  Let $\myalp$ be a finite alphabet and let $S\subseteq\myalpinf $ be some non-empty 
  set of infinite strings over  $\myalp$. 
  Given a perturbation function $X:\myalpinf \rightarrow \myalpall $ the 
  {\em  smoothed trie height} for $n$ strings over the set $S$ under the 
  perturbation function $X$, denoted by $H(S,n,X)$, is defined by
  \[
  H(S,n,x)=_\df    \max\limits_{\genfrac{}{}{0pt}{2}{A\subseteq S}{\|A\|=n}}
    \mymean{\max\limits_{s,t\in A}\lcp(X(s),X(t))}\mbox{.}
    \]
  \end{definition}
  Note that we assume that strings are perturbed independently.  
  For our smoothed analysis, the input set $S$ can either be arbitrary, i.e., 
  the above product space over all infinite strings from $\myalp$, or restricted. 
  We consider only the first variants, where the inputs are unconstrained.
  

\subsection{Perturbations by Probabilistic Finite Automata}
\label{sec:pert-prob-finite}
In this subsection we present our perturbation model which is based on probabilistic 
finite automata.

\paragraph{(Mealy-type) Probabilistic Finite Automata.}
\label{sec:mealy-type-pfas}
A probabilistic finite automaton~\cite{Paz-1971,Rabin-1963} is a standard way 
to model an unreliable deterministic system or a communication channel.
We suggest to consider random perturbation functions representable by 
probabilistic automata.
It is not our aim to develop a general theory of automata-based perturbation
functions. 
Instead, we use probabilistic finite automata as a compact, but nevertheless 
fairly general representation for string perturbation functions. 
We will define the probabilistic finite automata in a slightly non-standard way 
by separating input states from output states. 
This provides an easy way to describe automata computing non-length-respecting 
input-output relations.

A (Mealy-type) probabilistic finite automaton (PFA) over a finite alphabet 
$\myalp$ is a tuple $X=(R,W,\mu_R,\mu_W,\sigma)$ where:
\begin{itemize}
\item $R$ is a non-empty, finite set of {\em  input states}.
\item $W$ is a non-empty, finite set of {\em  output states}.
\item $\mu_R:R\times \myalp\times(R\cup W)\rightarrow [0,1]$ is the  
	{\em  transition probability function for input states} satisfying
  \[
  (\forall q\in R)(\forall a\in \myalp)~~\sum_{p\in R\cup W} \mu_R(q,a,p)=1\mbox{.}
  \]
  The semantics of the function $\mu_R$ is: if the PFA $X$ is in input state $q$ and the symbol
  $a$ is read, move into state $p$ with probability $\mu_R(q,a,p)$. Note that possibly $\mu_R(q,a,q)>0$.
\item $\mu_W:W\times \myalp\times(R\times W)\rightarrow [0,1]$ is the 
	{\em   transition probability function for output states} satisfying
  \[
  (\forall q\in W)(\forall a\in \myalp)~~\sum_{p\in R\cup W} \mu_W(q,a,p)=1\mbox{.}
  \]
  The semantics of the function $\mu_W$ is: if the PFA $X$ is in output state 
  $q$, with probability $\mu_W(q,a,p)$, write the symbol $a$ and move into state 
  $p$. Note that possibly $\mu_W(q,a,q)>0$.
\item $\sigma:R\cup W\rightarrow [0,1]$ is the {\em initial probability 
  distribution}, i.e., $\sigma$ satisfies $\sum_{q\in R\cup W}\sigma(q)=1$. 
\end{itemize}


We will identify with a PFA $X$ over the alphabet $\myalp$ a {\em  random mapping} 
$X:\myalpall \rightarrow\myalpall $, mapping finite of infinite strings to finite 
or infinite strings.
A {\em  computation of a PFA $X$ on an input symbol $a\in\myalp$} starts in some 
input state and stops when $X$ moves into an input state, again. 
The (possibly empty) output of the computation is composed by concatenating all 
output symbols of transitions leaving output states along which $X$ moved during 
the computation.
A computation of $X$ on an input string $t\in\myalpall $ is composed by the 
concatenation of the computations on the successive symbols of the string $t$, 
where the computation of $X$ on the symbol $t[i+1]$ starts in that input
 state in which the computation of $X$ on the symbol $t[i]$ stopped. 
The computation stops when $X$ reaches an input state and there is no more input 
symbol left to read. 
If $t$ is infinite, the computation never stops. 
The output of the computation is composed by concatenating all outputs of the 
computations on the individual symbols $t[1],t[2],\ldots$. 
A computation of $X$ is said to have {\em  output length} $m$ if the output has
length $m$ and is said to have {\em  input length} $l$ if it has read $l$ symbols 
of the input.


\paragraph{Edit Perturbations of Binary Strings.}
\label{sec:edit-pert-binary}

Edit operations, i.e., {\em substituting}, {\em deleting} or {\em inserting} 
symbols, are among the most fundamental operations for locally manipulating 
strings.  
Therefore, a smoothed analysis with respect to perturbation functions that
resemble these operations provide a better understanding of the good practical 
performance of tries.  
We say that a perturbation function on strings is an {\em edit perturbation} 
if it perturbs the input by randomly substituting, inserting or deleting symbols. 
Let $p,q\in(0,1)$.
The perturbation function $\SUB_p$ substitutes each symbol in 
the input string with its opposite symbol independently with probability $p$; 
the perturbation function $\INS_{pq}$ inserts before each symbol in the input 
string a number of $k$ symbols $a_1,\ldots,a_k$, where for $i\in\{1,\ldots, k\}$, 
$a_i$ equals $0$ with probability $q$ and $1$ with probability $1-q$. 
The number of inserted symbols is geometrically distributed with parameter $p$.
Finally, the perturbation function $\DEL_p$ reads the input string and deletes 
each symbols independently with probability $p$. 

Analyzing the smoothed trie height under each of the edit perturbations of binary
strings has been the starting point of our research in this field. 
It can be shown that the smoothed trie height under $\SUB_p$ and $\INS_{pq}$ is 
logarithmic and it is immediate that this does not hold for the function $\DEL_p$ 
because the input string $111\ldots$ is mapped to the output string $11\ldots$ 
deterministically.
For the convex combination of the edit perturbation matters are less trivial:
    let $p_S,p_I,q_I,p_D\in (0,1)$ be the respective parameters for the edit 
    perturbations and let $v=(v_S,v_I,v_D)\in[0,1]^3$ be such that $v_S+v_I+v_D=1$ 
    be the parameter vector for the convex combination.
    We say that a perturbation function $Y:\myalpall\rightarrow\myalpall$ 
    is the {\em  convex combination of the binary edit perturbations}, if $Y$ 
    can be represented by the PFA depicted in Figure~\ref{fig:pfa-convex}.

    
\begin{figure}[ht]
      \begin{center}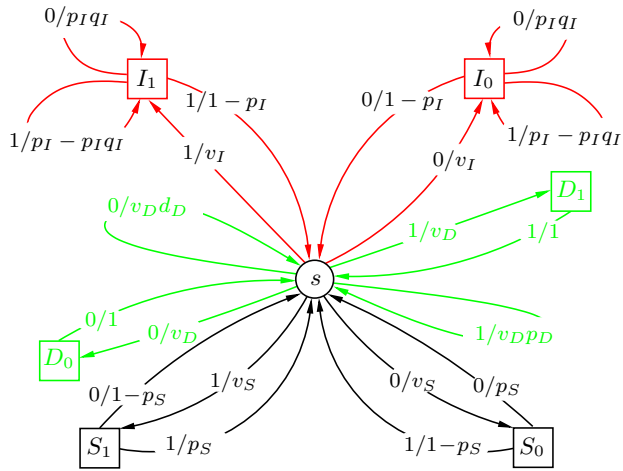\end{center}
      \vspace*{-10pt}
      \caption[Convex Combination PFA]{PFA $Y$ representing the convex combination 
		of the PFAs $\INS_{p_Iq_I}$, $\SUB_{p_S}$ and $\DEL_{p_D}$. 
		States in circles are input states and states in boxes are output states 
		and as usual transitions are only drawn if their probability is strictly 
		positive. Transitions are labeled by a tuples ``$a/x$'' where $a\in \myalp$ 
		and $0< x\le1$. The semantics is as follows: for a reading transition, 
		$a/x$ means ``if we read symbol $a$, then we move along the respective 
		transition with probability $x$''; for a writing transition, ``$a/x$'' 
		means ``with probability $x$, we move along the respective transition 
		and write $a$''. }
   \label{fig:pfa-convex}
 \end{figure}

 \paragraph{Star-like Perturbation Functions.}

All of the perturbations considered in the last section have in common that 
there is exactly one input state and that the computations on the
individual symbols never move between distinct output states. 
We now formally define a class of perturbation functions which are characterized 
by exactly these properties. 
Since their representation is a directed star graph with multi-edges and loops, 
where the unique input state is the center vertex, the set of output states is 
the set of terminal vertices, and the transitions having strictly positive 
probability gives the set edges, we call those PFAs and their respective
perturbation functions {\em  star-like}.

\begin{definition}
\label{def:star-like-pfa}
Let $\myalp$ be finite a alphabet and let $X=(R,W,\mu_R,\mu_W,\sigma)$ be a PFA over $\myalp$. 
$X$ is said to be {\em star-like} if the following hold:
\begin{itemize}
\item[(1)] $\|R\|=1$, i.e.,  $R=\{s\}$.
\item[(2)] The function $\mu_W$ is such that
\[
(\forall q,q'\in W,q\neq q')(\forall a\in\myalp)~~\mu_W(q,a,q')=0,
\]
i.e., the graph induced by the set $W$ and edge set $\{\{q,q'\}~:~(\exists a\in \myalp)\allowbreak
~\mu_W(q,a,q')>0\}$
consists of a number of connected components each of which is a single vertex.
\item[(3)] For all $q\in W$ it holds that $\sum_{ a\in\myalp}\mu_W(q,a,q)<1$, i.e., the probability that $X$ loops at
  $q$ is strictly less than one.
\end{itemize}
\end{definition}

Further, we consider a strict subclass of the star-like perturbation functions, namely the class of those perturbation
functions which are such that for each symbol $a\in\myalp$, there is exactly one output state, say $q_a$, that can be
reached from $s$ with positive probability when reading $a$. If additional to this the perturbation functions are 
{\em non-deleting}, i.e., there are no loops at $s$, then we say that they are {\em read-deterministic} perturbation
functions. Otherwise, i.e., there are symbols $a$ which are deleted with positive probability, we say that the
perturbation functions are {\em read-semi-deterministic}. It is easy to verify that all edit perturbations are star-like
perturbation functions and further that the functions $\INS_{pq}$ and $\SUB_p$ are read-deterministic and the function
$\DEL_p$ is read-semi-deterministic.

\begin{definition}
  \label{def:read-det-star-like-pfa}
  Let $\myalp$ be a finite alphabet and let $X=(\{s\},W,\mu_R,\mu_W,\sigma)$ be a star-like PFA over $\myalp$.
  $X$ is said to be {\em  read-semi-deterministic}, if for all $a\in\myalp$, there exist a constant $p_a\in[0,1]$ 
  and {\em  exactly one} output state $q_a$ such that $\mu_R(s,a,s)=p_a$ and  $\mu_R(s,a,q_a)=1-p_a$. 
  Further, $X$ is said to be {\em  read-deterministic}, if for all $a\in\myalp$, $p_a=0$, i.e., $\mu_R$ has {\em  no loops at $s$}.
\end{definition}

 \subsection{Comparison to Previous Random String Models}
 \label{sec:comp-prev-rand}

 In this subsection we compare our semi-random string model to purely random 
 string models. 
 One property that the sequences from most random sources possess is the 
 {\em  mixing property}, which as we mentioned implies that R{\'{e}}nyi's 
 Entropy of second order, i.e., the limit~\eqref{eq:RenyiEntropy}, exists. 
 We show that these assumptions do not hold in general for sequences which 
 result from the perturbation of a non-random input sequence by means of a 
 star-like perturbation function. 
 To this end, let $X$ be a read-semi-deterministic PFA such that for two 
 distinct symbols $a$ and $b$ it holds that $Q_a\neq Q_b$, where for 
 $i\in\{a,b\}$,
 \[
 Q_i=(\mu_W(q_i,a_1,q_i)+\mu_W(q_i,a_1,s),\ldots,\mu_W(q_i,a_r,q_i)+\mu_W(q_i,a_r,s))\mbox{.}
 \]
 It is easy to verify that the output of the pairs $(aaa\ldots,X)$ and 
 $(bbb\ldots,X)$, respectively, have the same probability distributions as 
 memory-less random sources with parameter vectors $Q_a$ and $Q_b$, respectively.
 Then a standard calculation (cf.~\cite{Szpankowski-2001}) gives the following: For $i\in\{a,b\}$ the limit depends on the input string:
 \[
 \lim_{n\rightarrow\infty} \frac{-\ln{\sum_{\alpha\in\myalple{n}}\myprob{X(iii\dots)[1\ldots n]=\alpha}^2}}{2n} = Q_i 
\]
The enables us to give lower bounds on the smoothed trie height. 
\begin{proposition}
  \label{prop:lower-bound}
  Let $X=(\{s\},W,\mu_R,\mu_W,\sigma)$ be a read-semi-deterministic PFA over a 
  finite alphabet $\myalp$ in canonical form (for a definition see below) and 
  let 
 $P=\max_{a\in\myalp}\sum_{b\in\myalp} (\mu_W(q_a,b,q_a)+\mu_W(q_a,b,s))^2
  $.
  Then for all $\varepsilon>0$,
  \[
  H(\myalpinf ,n,X) \ge 2(1-\varepsilon)\log_{1/P}{n}-o(1)\mbox{.}
  \]
\end{proposition}

\section[Main Result: Star-like Perturbation Functions]{Main Result:
Smoothed Trie Height under Star-like Perturbation Functions}
\label{sec:the-result}

\paragraph{A Dichotomous Result.}
In this section we present the main result of this work. 
Let $X$ be star-like and let $X=(\{s\},W,\mu_R,\mu_W,\sigma)$ be the representing PFA.
To ease the analysis we assume that perturbations start in the input state $s$ with 
probability one, i.e., that $\sigma(s)=1$ and for all $q\in W$, $\sigma(q)=0$ holds, 
and we say that such a perturbation function is in {\em canonical form}. 
The following dichotomous result for star-like perturbation functions 
over arbitrary input sets can be proven.

 \begin{theorem}
   \label{thm:star-like}
   Let $X$ be a star-like string perturbation function over a finite alphabet $\myalp$ in canonical form, 
   represented by the PFA $X=(\{s\},W,\mu_R,\mu_W,\sigma)$ such that for all $a\in\myalp$ it holds that $\mu_R(s,a,s)<1$.
   Then the following statements are equivalent.
   \begin{enumerate}
   \item[$(1)$] $(\forall a,b\in\myalp)~~    \mu_R(s,a,s)+\sum_{q\in W}\mu_R(s,a,q)\cdot(\mu_W(q,b,q)+\mu_W(q,b,s))<1$
   \item[$(2)$]  $H(\myalpinf,n,X)\in O(\log{n})$.
   \end{enumerate}
 \end{theorem}
 
 Before we discuss the meaning of the above theorem, we note that it directly 
 yields the following corollary concerning the smoothed trie height under convex 
 combinations of edit perturbations of arbitrary binary strings.
 
 \begin{corollary}
  Let $p_S,p_I,q_I,p_D\in (0,1)$ and let $v=(v_S,v_I,v_D)\in[0,1]^3$ be such 
  that $v_S+v_I+v_D=1$ and let $Y$ be string perturbation function which is 
  computed by the PFA depicted in Figure~\ref{fig:pfa-convex}. 
  Then, $H(\{0,1\}^\omega,n,Y)\in O(\log{n})$ if and only if $v_D<1$.
  In other words, the smoothed trie height is logarithmic if and only if 
  the convex combination of edit perturbations does not collapse to deletions.
\end{corollary}

 In general, statement~$(1)$ of the theorem gives a set of necessary and 
sufficient conditions such that the smoothed trie height $H(\myalpinf,n,X)$ 
is logarithmic in $n$ if those conditions are satisfied and unbounded, 
otherwise. 
These conditions are especially appealing, because they can be verified  
easily and efficiently by looking at the transition probability function 
of the representing PFA. 
For general star-like perturbation functions the verification can be done 
algorithmically in time $O(\|\myalp\|^2\cdot\|W\|)$. 
 Note that the {\em  additional constraint} regarding the deletion probabilities, 
 i.e., that for all $a\in\myalp$ it holds that $\mu_R(s,a,s)<1$ cannot be 
 dropped: let $a\in\myalp$ be such that $\mu_R(s,a,s)=1$ and let $t=aaa\ldots$. 
 Then $X(a)=\epsilon$ with probability one and it becomes obsolete to speak of 
 smoothed trie height in this particular case.
  
\paragraph{Quantitative Analyses.}
  When performing a smoothed analysis it is usual to quantify the influence of 
  the parameters of the perturbation function on the smoothed complexity of a 
  problem. 
  We can give the following quantitative result on the smoothed trie height.  
  Let $X(\{s\},W,\mu_R,\mu_W,\sigma)$ be a star-like PFA over the finite alphabet
  $\myalp$ in canonical form.
  For $q\in W$ the {\em return probability from state $q$} is defined as
  \[
  \eta_q =_\df \sum_{a\in\myalp} \mu_W(q,a,s)\mbox{.}
  \]
  Also, for the sake of exposition, define for $a\in\myalp$ and $q\in W$
  \[
  \myread{a}{q} =_\df \mu_R(s,a,q)~~~~~\text{ and}~~~~~
  \mydel{a} =_\df \mu_R(s,a,s)\mbox{.}
  \]
  \begin{theorem}
    \label{thm:tight-for-star-like} 
    Let $X$ be a star-like string perturbation function over a finite alphabet $\myalp$ in canonical form, represented
    by the PFA $X=(\{s\},W,\mu_R,\mu_W,\sigma)$ such that for all $a\in\myalp$ it holds that $p_a<1$ and such that \[
    (\forall a,b\in\myalp)~~ \mu_R(s,a,s)+\sum_{q\in W}\mu_R(s,a,q)\cdot(\mu_W(q,b,q)+\mu_W(q,b,s))<1\mbox{,}\] where we
    denote the maximum term by $\delta$.  Let $\gamma = 1/\tilde{z}$, where $\tilde{z}$ is the pole of minimum modulus
    of the function
   \[
     \tilde{\mathcal{Z}}_X(z)
     =\prod_{i=1}^{r}\left(
         1-
         \delta\cdot\mydel{a_i} - 
         \sum\limits_{j=1}^{v} \frac{\delta\cdot\myread{a_i}{q_j}\cdot\eta_{q_j}\cdot z}{1-(1-\eta_{q_j})\cdot z}
       \right)^{-1}   
   \]
   Then, for $n$ sufficiently large and for all $\varepsilon>0$ it holds that 
\[
   H(\myalpinf ,n,X)\le 2\cdot\lceil(1+\varepsilon)\log_{1/\gamma}{n}\rceil + o(1).
\]
\end{theorem}
 Note that for the case of read-semi-deterministic perturbation functions we 
 also get a lower bounds from Proposition~\ref{prop:lower-bound}.
 Unfortunately, this lower bound does not match our quantitative upper bound.
 Non-matching upper and lower bounds can also be found in the dichotomous result 
 of Devroye~\cite{Devroye-1984} .

\section{Conclusions}
There are two main open problems posed by this paper: the first concerns the 
extension of our perturbation functions to more general string perturbation 
functions which can be represented by PFAs. 
Clearly, general PFAs which can model real-world string sources such as sensors 
more appropriately are one possible extension. 
We are particularly interested in \emph{probabilistic push-down automata} 
because they provide a way to model \emph{random transpositions}, which occur
quite frequently in non-random data such as DNA sequences. 
The second open problem concerns the smoothed analysis of other parameters and 
related data structures under our model. 
Particularly, we actually try to analyze the smoothed trie height of suffix 
trees. 
There, it is believed that the mixing condition is a necessary ingredient to 
prove logarithmic smoothed trie height (cf.~\cite{Szpankowski-1993a}). 
Since our model does not satisfy the mixing condition, a positive result would 
give new insights in the practical performance of such data structures.


{\small

}

\newpage

\appendix
\section{Proof of Proposition~\ref{prop:lower-bound}}
\begin{proof}
  Let $X=(\{s\},W,\mu_R,\mu_W,\sigma)$ be a read-semi-deterministic PFA over a finite alphabet $\myalp=\{a_1,\ldots,a_r\}$ in canonical
  form. Let $A$ be a set of $n$ infinite strings each of which starts with $2n$ repititions the symbol
  $a\in\myalp$ such that
  \[
  P=\max_{a\in\myalp}\left(\sum_{b\in\myalp} (\mu_W(q_a,b,q_a)+\mu_W(q_a,b,s))^2\right)
  \]
  is maximal. It holds that
  \[
  H(\myalpinf ,n,X)\ge 
  \mymean{\max_{s,t\in A} \lcp(X(s),X(t))}\mbox{.}
  \]
  Let $k = 2\lceil\log_{1/P}{n}\rceil$. 
  Now, for each string $s\in A$ the probability that \[|X(s[1\ldots 2 n])|\ge k,\]
  i.e., the computations of $X$ on the prefix of $s$ of input length $2 n$ has length at least $k$ satisfies
  \begin{eqnarray*}
    \myprob{|X(s[1\ldots 2 n])|\ge k}&=&1-    \myprob{|X(s[1\ldots 2 n])|< k}\\
    &\ge&1-\myprob{\text{at least $2n-k$ symbols are deleted}}\\
    &=&1-\sum_{i=0}^{k-1}\binom{2 n}{i}\cdot (p_a)^{2n-i}\\
    &\ge& 1-k\cdot \binom{2 n}{k}\cdot (p_a)^{2 n-k} \\
    &=&1-o((p_a)^n)\mbox{.}
  \end{eqnarray*}
  Now, with probability $1-o((p_a)^n)$ the prefix of length $k$ of output of the computation of $X$ on $s$
  has the same distribution as the prefix of an string that is written by a Memory-less random source with parameter
  vector $p\in(0,1)^r$, where for $i\in\{1,\ldots, r\}$, $p_i=\mu_W(q_a,a_i,q_a)+\mu_W(q_a,a_i,s)$.
  For such a source and two random strings $s',s''$ it holds for every $k\in\Nat_+$ that
  \[
  \myprob{\lcp(s',s'')\ge k} = P^k\mbox{.}
  \]
  Let $A=\{s_1,\ldots,s_n\}$ and $k=2(1-\varepsilon)\log_{1/P}{n}$ for $\varepsilon>0$ and  for $i,j\in\{1,\ldots,r\}$ let $C_{ij}=\lcp(s_i,s_j)$.
  From the preceeding,
  \[
  \myprob{C_{ij}\ge k } = P^k\cdot (1-o((p_a)^n)^2\mbox{.}
  \]
  This holds particularly, because $k$ is exponentially smaller than $n$, i.e., $k = 2\lceil\log_{1/P}{n}\rceil$.
  Let $H$ be the height of a trie which is build over the set $A$.
  Using the Second Moment Method, the following claim can be shown.
  \begin{claim}{{\em [see Section~$4.2.3$ in~\cite{Szpankowski-2001}] }}
    Under the above conditions, for any $\varepsilon>0$ it holds that
    \[
    \myprob{H> 2(1-\varepsilon)\log_{1/P}{n} } =1 - O(1/n^\varepsilon)\mbox{.}    
    \]
  \end{claim}
  The above claim  implies that for every $\varepsilon>0$
  \[
  H(\myalpinf,n,X)
  \ge 
  \mymean{\max_{s,t\in A} \lcp(X(s),X(t))}    
  =  \mymean{H} \ge
  2(1-\varepsilon)\log_{1/P}{n}-o(1)\mbox{.}
  \]
  This proves the Theorem
\end{proof}


\section{Overview on the Proofs of Theorem~\ref{thm:star-like} and Theorem~\ref{thm:tight-for-star-like}}
\label{sec:smoothed-trie-height-star-like}

In this section we give an overview on the proofs of Theorem~\ref{thm:star-like} Theorem~\ref{thm:tight-for-star-like}.  
The details of the proofs are given in the subsequent sections.

First, we show that $H(S,n,X)$ grows at most as
$2\log_{1/\gamma}{n}$, if the {\em coincidence probability of length $m$} of two independent perturbations of the same
string $s\in S$, i.e., $\myprob{\lcp(X(s),X(s))\ge m}$, can be bounded from above by $\gamma^m$ for some $\gamma<1$.
The following lemma holds for arbitrary string perturbation functions.  

Its formal proof can be found in Section~\ref{sec:proof-lemma-tailbound}.

  \begin{lemma}
    \label{lem:tailbound}
    Let $\myalp$ be a finite alphabet and let $m_0\in\Nat$ and $\gamma\in\Real$ satisfying $0<\gamma<1$. Let
    $X:\myalpinf \rightarrow \myalpall $ be a perturbation function and let $S\in\myalpinf $ be a non-empty
    set of infinite strings. Let $n>\gamma^{-m_0/2}$. If there is a polynomial $\Pi(z)$ of fixed degree $d\in\Nat$, such that
    for all $s\in S$ and all $m\ge m_0$ it holds that the coincidence probability of two independent perturbations of $s$
    satisfies
    \[
    \myprob{\lcp(X(s),X(s))\ge m}\le \Pi(m)\cdot\gamma^m,
    \]
    then for all $\varepsilon>0$ it holds that
    $
    H(S,n,X) \le 2\cdot\lceil(1+\varepsilon) \log_{1/\gamma}{n}\rceil + o(1)$.
  \end{lemma}

 \begin{proof}[Proof of Theorem~\ref{thm:star-like}]
 Let $X=(\{s\},W,\mu_R,\mu_W,\sigma)$ be a star-like PFA over the alphabet 
 $\myalp=\{a_1,\ldots,a_{r}\}$ in canonical form.
 In order to prove the equivalence of the two statements, we claim that 
 $(2)\Rightarrow(1)$ and that $(1)\Rightarrow(2)$. 
 Then, the theorem follows. 
 The first claim, i.e., that $(2)\Rightarrow(1)$, can easily be established by 
 contraposition.
   
 \begin{claim}
 In the setting of Theorem~\ref{thm:star-like}, it holds that $(2)\Rightarrow(1)$.
 \end{claim}
 
 \begin{proof}
 We prove the claim by contraposition: to this end assume that $(1)$ does not hold, 
  i.e., there are symbols $a,b\in\myalp$ such that 
   \[
   \mu_R(s,a,s)+\sum_{q\in W}\mu_R(s,a,q)\cdot(\mu_W(q,b,q)+\mu_W(q,b,s)=1\mbox{.}
   \]
  Thus $\myprob{b\sqsubseteq X(a)} = 1-\mu_R(s,a,s)$. 
  Let $t=aaa\ldots$ and let $s=bbb\ldots$. 
  Then $X$ maps $t$ to $s$ with probability one. 
  Therefore, $H(\myalpinf,n,X)$ is unbounded. 
  The claim follows.
 \end{proof}

 The second claim is less easy to prove: in order to show that $(1)\Rightarrow(2)$, 
 we prove that $(1)$ is a sufficient condition such that the tail-bound 
 (Lemma~\ref{lem:tailbound}) can be applied. 
 Particularly, we show that under the assumption that~$(1)$, for arbitrary 
 $t\in\myalpinf $ and $m\in\Nat_+$ sufficiently large there are suitable positive 
  constants $u,v,\gamma$ satisfying $0<\gamma<1$ such that
 \[
  \myprob{\lcp(X(t),X(t)\ge m} = \sum_{\alpha\in\myalple{m} } \myprob{\alpha\sqsubseteq X(t)}^2
   \le ( u\cdot m + v )\cdot\gamma^m\mbox{.}
 \]
 To this end, for $\alpha,t\in\myalpfin$ let $\toalpxy{X}{\alpha}{t}$
 be the probability that a computation of $X$ on $t$ that has input length 
 $|t|$ has the prefix $\alpha$.
 Then for $t\in\myalpinf $ and $\alpha\in\myalpfin $ we have the following 
 identity
 \begin{equation}
   \nonumber
   \myprob{\alpha\sqsubseteq X(t)} = \sum_{l=1}^{\infty}\toalpx{X}{\alpha}{t}{1}{l}
 \end{equation}
 and thus for $m\in\Nat_+$ we have
 \begin{equation}
   \nonumber   \label{eq:12a}
   \myprob{\lcp(X(t),X(t)\ge m} =  
		\sum_{\alpha\in\myalple{m} }\Biggl(\sum_{l=1}^{\infty}\toalpx{X}{\alpha}{t}{1}{l}\Biggr)^2\mbox{.}
 \end{equation}
 Next we split the right-hand side of the above equation into two suitable parts 
 by an application of Cauchy's Inequality: let $d\in\Real_+$ be a constant to be 
 defined in a moment. 
 Then
 \begin{eqnarray}\nonumber
 \lefteqn{  \sum_{\alpha\in\myalple{m} }\Biggl(\sum_{l=1}^{\infty}\toalpx{X}{\alpha}{t}{1}{l}\Biggr)^2}\\
   &=&
   \sum_{\alpha\in\myalple{m} }\Biggl(\sum_{l=1}^{\lceil d \cdot m \rceil}\toalpx{X}{\alpha}{t}{1}{l}
   +\sum_{l=\lceil d \cdot m \rceil+1}^{\infty}\toalpx{X}{\alpha}{t}{1}{l}\Biggr)^2 ~~~~~~~~~~~~~~~
    \\\label{eq:13}
   &\le&
   2\cdot\sum_{\alpha\in\myalple{m} }\Bigl(\sum_{l=1}^{\lceil d \cdot m \rceil}\toalpx{X}{\alpha}{t}{1}{l}\Bigr)^2
   +
   2\cdot\sum_{\alpha\in\myalple{m} }\Bigl(\sum_{l=\lceil d \cdot m \rceil+1}^{\infty}\toalpx{X}{\alpha}{t}{1}{l}\Bigr)^2\mbox{.}
 \end{eqnarray}
 Then, we prove an exponentially decreasing upper bound on each of the two addends 
 in~\eqref{eq:13} under the assumption that~$(1)$. 
 To this end, we define for $m\in\Nat_+$, $d\in\Real_+$ and $t\in\myalpinf $:
 \begin{eqnarray*}
 \Phi(t,m,d) &=_\df& 2\cdot\!\!\sum_{\alpha\in\myalple{m} }\Bigl(\sum_{l=1}^{\lceil d \cdot m \rceil}\toalpx{X}{\alpha}{t}{1}{l}\Bigr)^2 \\
 \Psi(t,m,d) &=_\df& 2\cdot\!\!\sum_{\alpha\in\myalple{m} }\Bigl(\sum_{l=\lceil d \cdot m \rceil+1}^{\infty}\toalpx{X}{\alpha}{t}{1}{l}\Bigr)^2
  \mbox{.}
 \end{eqnarray*}
 \begin{claim}
   \label{claim:boundonpart1}
   Let $d\in\Real_+$ be fixed and let $\gamma =
    1/\tilde{z}$, where $\tilde{z}$ is the pole of minimum modulus of the function
   \[
     \tilde{\mathcal{Z}}_X(z)
     =\prod_{i=1}^{r}\left(
         1-
         \delta\cdot\mydel{a_i} - 
         \sum\limits_{j=1}^{v} \frac{\delta\cdot\myread{a_i}{q_j}\cdot\eta_{q_j}\cdot z}{1-(1-\eta_{q_j})\cdot z}
       \right)^{-1}   \mbox{.}
       \]Under the assumption that~$(1)$, there is polynomial  $\Pi(z)$ of fixed degree $\le r$
       such that $ \Phi(t,m,d) \le \Pi(m)\cdot \gamma^m$.
 \end{claim}
 
 \begin{claim}
   \label{claim:boundonpart2}
   Let $d\in\Real_+$ be fixed.  For a star-like perturbation function $X$ as in the setting of
   Theorem~\ref{thm:star-like}, there exist constants $c,\gamma_2\in\Real$ satisfying $\gamma_2<1$ and such that $
   \Psi(t,m,d) \le c\cdot \gamma_2^m$.
 \end{claim}
\noindent
The detailed proofs of the two claims can be found in Section~\ref{sec:detailed-proof-of-thm}. 
We now fix $d$. The above directly yields
 \[
 \myprob{\lcp(X(t),X(t))\ge m} \le  \Psi(t,m,d) + \Phi(t,m,d) \le (\Pi(m)+c)\cdot\tilde{\gamma}^m
 \]
 for $\tilde{\gamma}=\max\{\gamma_1,\gamma_2\}<1$. Thus we may apply the tail-bound.
Together this shows the sought-after claim.
\begin{claim}
In the setting of Theorem~\ref{thm:star-like}, it holds that $(1)\Rightarrow(2)$.
\end{claim}
This proves Theorem~\ref{thm:star-like}.
 \end{proof}
\noindent
Note that it can also be shown (see Appendix~\ref{sec:detailed-proof-of-thm}) 
that for $d$ sufficiently large,
$\lim_{m\rightarrow\infty } \Phi(t,m,d)/\Psi(t,m,d)=0$
which implies Theorem~\ref{thm:tight-for-star-like}.

\section{Proof of Lemma~\ref{lem:tailbound}}
\label{sec:proof-lemma-tailbound}
  \begin{proof}
 Let $S$ be a non-empty set of infinite strings over a finite alphabet $\myalp$.
 Let $\varepsilon>0$ and let $k\in\Nat_+$ be arbitrary. Then
  \begin{eqnarray}\nonumber
    H(S,n,X) &=&
     \max\limits_{\genfrac{}{}{0pt}{2}{A\subseteq S}{\|A\|=n}}
     \mymean{\max\limits_{s,t\in A}
       \lcp(X(s),X(t))}\\\nonumber
       &=& \max\limits_{\genfrac{}{}{0pt}{2}{A\subseteq S}{\|A\|=n}}\sum_{i=1}^{\infty}\myprob{\max\limits_{s,t\in A}
       \lcp(X(s),X(t))\ge i}\\\nonumber
     &\le& \sum_{i=1}^{\infty}\max\limits_{\genfrac{}{}{0pt}{2}{A\subseteq S}{\|A\|=n}}\myprob{\max\limits_{s,t\in A}
       \lcp(X(s),X(t))\ge i}\\\label{eq:prof:splitting-lemma:1}
      &\le& k
     + \sum_{i=k+1}^{\infty}\max\limits_{\genfrac{}{}{0pt}{2}{A\subseteq S}{\|A\|=n}}\myprob{\max\limits_{s,t\in A}
       \lcp(X(s),X(t))\ge i}\\\label{eq:prof:splitting-lemma:2}
     &\le& k+  n^2 \cdot\sum_{i=k+1}^{\infty}\max\limits_{s,t \in S}\myprob{\lcp(X(s),X(t))\ge i}\mbox{.}
   \end{eqnarray}
   Inequality~\eqref{eq:prof:splitting-lemma:2} follows from Boole's Inequality and~\eqref{eq:prof:splitting-lemma:1}
   holds, because the in sum of probabilities each addend of the first $k$ addends can by bounded by one.
   Now we expand each addend of the right-hand side  and apply Cauchy's Inequality in its standard from : 
   \begin{eqnarray}\nonumber
   \max\limits_{s,t \in S}\myprob{\lcp(X(s),X(t))\ge i}
   &=& \max\limits_{s,t\in S}\sum_{\alpha\in\myalp^i}\myprob{\alpha\sqsubseteq X(s)}\cdot\myprob{\alpha\sqsubseteq X(t)}\\\nonumber
      &\le& \max\limits_{s,t\in S}\sqrt{\sum_{\alpha\in\myalp^i}\myprob{\alpha\sqsubseteq X(s)}^2}\cdot\sqrt{\sum_{\alpha\in\myalp^i}\myprob{\alpha\sqsubseteq X(t)}^2}\\\nonumber
      &\le& 
      \max\limits_{s\in S} \sum_{\alpha\in\myalp^i}\myprob{\alpha\sqsubseteq X(s)}^2  \\\nonumber
      &=&\max\limits_{s\in S}
      \myprob{\lcp(X(s),X(s))\ge i}\mbox{.}
 \end{eqnarray}  
 Now, we have that for all $k\in\Nat_+$ it holds that
 \[
 H(S,n,X) \le k + n^2\cdot \sum_{i=k+1}^{\infty}\max_{s\in S} \myprob{\lcp(X(s),X(s))\ge i}\mbox{.}
 \]
 Let $d\in\Nat_+$ and let $\Pi(z)$ be a polynomial of degree $d$ such that the assumption of the theorem holds. 
 Set $k=2\cdot\lceil(1+\varepsilon)\log_{1/\gamma}{n}\rceil\ge m_0$. Then
 \begin{eqnarray*}
 H(S,n,X) 
 &\le& k + n^2\cdot \sum_{i=k+1}^{\infty}\max_{s\in S} 
 \myprob{\lcp(X(s),X(s))\ge i}\\
 &\le& 2\cdot\lceil(1+\varepsilon)\log_{1/\gamma}{n}\rceil +
 \sum_{i=2\cdot \lceil(1+\varepsilon)\log_{1/\gamma}{n}\rceil+1}^{\infty} \Pi(i)\cdot n^2\cdot \gamma^{i}
 \end{eqnarray*}
 It is easy to see that the latter term is in $o(1)$:
\begin{eqnarray*}
\lefteqn{  \sum_{i=2\cdot\lceil(1+\varepsilon)\log_{1/\gamma}{n}\rceil+1}^{\infty} \Pi(i)\cdot n^2\cdot
    \gamma^{i}}~~~~~~~~~~~~~~~~~~~~~~~\\
&=&
\sum_{i=1}^{\infty} \Pi(2\cdot\lceil(1+\varepsilon)\log_{1/\gamma}{n}\rceil+i )\cdot n^2\cdot
\gamma^{2\cdot\lceil(1+\varepsilon)\log_{1/\gamma}{n}\rceil}\cdot\gamma^{i}\\
  &\le&
\sum_{i=1}^{\infty} \Pi(2\cdot\lceil(1+\varepsilon)\log_{1/\gamma}{n}\rceil+i )\cdot n^2\cdot n^{-2-2\varepsilon}\cdot\gamma^{i}\\
&=&
n^{-2\varepsilon}\cdot\sum_{i=1}^{\infty} \Pi(2\cdot\lceil(1+\varepsilon)\log_{1/\gamma}{n}\rceil+i )\cdot
\gamma^{i}\in o(1)\mbox{.}
\end{eqnarray*}
This proves the lemma.
\end{proof}
 
\section{Detailed Proof of Theorem~\ref{thm:star-like}}
\label{sec:detailed-proof-of-thm}
 
Before we actually start with proving the first claim,
we first show how to express the term $\sum_{\alpha\in\myalple{m} }\toalpx{X}{\alpha}{t}{1}{l}$  subject to the transition
probabilities of $X$. Then we establish Claim~\ref{claim:boundonpart1}  in Section~\ref{sec:bounding-phi}. Afterwards, we turn to the
proof of Claim~\ref{claim:boundonpart2}  in Section~\ref{sec:bounding-psi}. 
This then proves Theorem~\ref{thm:star-like}.

\subsection{Prerequisites: computations of star-like PFAs}
\label{sec:comp-star-like}
In this section, we prove the following Lemma which will be one of the important ingredients in the proofs of
Claims~\ref{claim:boundonpart1} and~\ref{claim:boundonpart2}. In particular, the lemma gives a (nearly exact)
expression of the term $\sum_{\alpha\in\myalple{m} }\toalpx{X}{\alpha}{t}{1}{l}$, i.e., the probability that a computation of input length
$l$ on the prefix of $t$ has output length at least $m$, subject to the transition
probabilities of $X$.

\begin{lemma}
  \label{lem:expanding-the-sum}
  Let $X=(\{s\},W,\mu_R,\mu_W,\sigma)$ be a star-like PFA over the finite alphabet $\myalp$ in canonical form
  and let $t\in\myalpinf $ and $l\in\Nat_+$. Let $f:\Nat\rightarrow \{0,1\}$ the following defined function: for $x\in\Nat$,
     \[
     f(x)=_\df
     \left\{
       \begin{array}{ll}
         1 & \text{, if $x =0$}\\
        0 & \text{, otherwise.}
      \end{array}
    \right.
    \]
    The function $f$ is used to indicate deleted symbols in the computations of $X$. Then
    \begin{eqnarray*}
    \lefteqn{\sum_{\alpha\in\myalple{m} }\toalpx{X}{\alpha}{t}{1}{l}}~~~~~~~~~~
    \\
    &\le& \frac{1}{\tilde{\eta}}\cdot\!\!\!\! 
      \sum_{m_1+\ldots+m_l=m}      
      \prod_{i=1}^{l}
      \Bigl(\!
      f(m_i)\!\cdot\!\mydel{t[i]}\!+\!
      (1\!-\!f(m_i))\!\cdot\!\!\sum_{q\in W}\myread{t[i]}{q}\eta_q(1\!-\!\eta_q)^{m_i-1}
      \!\Bigr)\mbox{,}
      \end{eqnarray*}
    where $\tilde{\eta}=\min_{q\in W}\eta_q$ denotes the {\em  minimum return probability}.
\end{lemma}

\ifthenelse{\value{prooftype}=1}{
  \begin{proof}
    Recall that the term which we seek to bound is the probability that the computation of $X$ on $t$ of input length $l$
    has output length at least $m$. Since $X$ is star-like and given in canonical form,
    each computation of $X$ on $t$ starts in the input state and then 
    moves into some output state, from which it writes the output, before it moves into the input state again, where
    it reads the next symbol of the input and continues the computations as described above. Thus, each computation can be
    decomposed into the computations on the successive individual symbols of $t$. For a computation of input length $l$ and
    output length $m$, there are $\binom{m+l-1}{l-1}$ possibilities to concatenate $l$ such computations on individual
    symbols such that they give a computation of output length $m$: this equals the number of decompositions of $m$ into $l$
    non-negative addends. Note, that addends might be equal to zero, because computations might have output length zero,
    The computations on the first $l-1$ input symbols {\em  must} return into the input state, whereas the
    computation on the $l$-th and last input symbol may either loop at its output state or return back into the input state after having
    written the $m$-th and last symbol of the output.\\

    Now, consider a fixed decomposition $m_1+\ldots +m_l=m$ into
    possibly empty computations. For $i\in\{0,\ldots,l\}$, if $m_i=0$ then the probability that the computation of $X$
    on the symbol $t[i]$ has output length zero is 
    \begin{equation}
      \label{eq:16a}
      \myprob{X(t[i])=\epsilon} = \mydel{t[i]}\mbox{.}
    \end{equation}
    For $i\in\{1,\ldots,l-1\}$, if $m_i>0$ then the probability that the computation of $X$
    on the symbol $t[i]$ has output length {\em  exactly} $m_i$ is equal to
    \begin{equation}
      \label{eq:16b}
      \sum_{\alpha\in\myalple{m_i}}\myprob{X(t[i])=\alpha }= \sum\limits_{q\in W} \myread{a}{q}\cdot\eta_q\cdot(1-\eta_q)^{m_i-1}
    \end{equation}
    and the probability that the computation of $X$  on the symbol $t[l]$ has output length {\em  at least} $m_l>0$ is equal to
    \begin{equation}
      \nonumber
      \sum_{\alpha\in\myalple{m_l}}\myprob{\alpha\sqsubseteq X(t[l])}=
      \sum\limits_{q\in W} \myread{a}{q}\cdot(1-\eta_q)^{m_l-1}\mbox{.}      
    \end{equation}
    Let $\tilde{eta}=\min_{q\in W} \eta_q$ be the minimum return probability. The term~\eqref{eq:16c} can be bounded as

    \begin{equation}
      \label{eq:16c}
      \sum\limits_{q\in W} \myread{a}{q}\cdot(1-\eta_q)^{m_l-1}
    \le\frac{1}{\tilde{\eta}}\cdot\sum\limits_{q\in W} \myread{a}{q}\cdot\eta_q\cdot(1-\eta_q)^{m_l-1}\mbox{.}
  \end{equation}
  Now, using the indicator function $f$ to choose the correct term for $i\in\{1,\ldots,l\}$, i.e., the
    term~\eqref{eq:16a}, if $m_i=0$ and the term~\eqref{eq:16b} if $m_i>0$ and $i<l$ or the term~\eqref{eq:16c} if 
    $m_l>0$ and $i=l$,  the probability that the computation of $X$ on $t$ of input length $l$ that can be decomposed as
    $m_1+\ldots +m_l=m$ has length at least $m$ is can be bounded by the product
    \[
    \frac{1}{\tilde{\eta}}
    \cdot
    \prod_{i=1}^{l}\Bigl(
    f(m_i)\cdot \mydel{t[i]}+(1-f(m_i))\cdot\sum\limits_{q\in W} \myread{a}{q}\cdot\eta_q\cdot(1-\eta_q)^{m_i-1}\Bigr)\mbox{.}      
    \]
    Summing over all possible decompositions, we get
    \begin{eqnarray*}
    \lefteqn{\sum_{\alpha\in\myalple{m} }\toalpx{X}{\alpha}{t}{1}{l}}~~~~~~~~~~\\
    &\le&\frac{1}{\tilde{\eta}}\cdot\!\!\!\! 
    \sum_{m_1+\ldots+m_l=m}      
    \prod_{i=1}^{l}
    \Bigl(\!
    f(m_i)\!\cdot\!\mydel{t[i]}\!+\!
    (1\!-\!f(m_i))\!\cdot\!\!\sum_{q\in W}\myread{t[i]}{q}\eta_q(1\!-\!\eta_q)^{m_i-1}
    \!\Bigr)
    \end{eqnarray*}
    which proves the lemma.
  \end{proof}
}{
Since $X$ is star-like and given in canonical form,
each computation of $X$ on some infinite input string starts in the input state and then 
moves into some output state, from which it writes the output, before it moves into the input state again, where
it reads the next symbol of the input and continues the computations as described above. Thus, each computation can be
decomposed into the computations on the successive individual symbols of the input. For a computation of input length $l$ and
output length $m$, there are $\binom{m+l-1}{l-1}$ possibilities to concatenate $l$ such computations on individual
symbols such that they give a computation of output length $m$: this equals the number of decompositions of $m$ into $l$
non-negative addends. Note, that addends might be equal to zero, because computations might have output length zero.
Also, note that the computations on the first $l-1$ input symbols {\em  must} return into the input state. The
computation on the $l$-th  and last input symbol may either loop at its output state or return back into the input state after having
written the $m$-th and last symbol of the output. The above is captured in the following proposition.

\begin{proposition}
  \label{pro:decomposition1}
    Let $X=(\{s\},W,\mu_R,\mu_W,\sigma)$ be a star-like PFA over the finite alphabet $\myalp$ in canonical form
  and let $t\in\myalpinf $ and $k,l\in\Nat_+$ and $\alpha\in\myalpfin $. Then
\[
\toalpx{X}{\alpha}{t}{k}{l}=
\left\{
  \begin{array}{ll}
    1 & \text{, if $|\alpha|=0$}\\
    0 & \text{, $|\alpha|>0$}\\
    & \text{~~and $k>l$}\\
    \sum\limits_{\alpha=\beta_0\ldots \beta_{l-k}}\prod\limits_{i=0}^{l-k-1}
    \myprob{X(t[k\!+\!i])=\beta_i}\myprob{\beta_{l-k}\sqsubseteq X(t[l])}
    & \text{, otherwise.}
  \end{array}
\right.
\] 
\end{proposition}

We consider a computation of $X$ on an individual input symbol.
\begin{proposition}
\label{prop:decomp2}
  Let $X=(\{s\},W,\mu_R,\mu_W,\sigma)$ be a star-like PFA over the finite alphabet $\myalp$ in canonical form and let
  $m\in\Nat$. Then for $\alpha\in\myalpfin $ and $a\in\myalp$,
  \begin{equation}\label{eq:22}
  \sum_{\alpha\in\myalple{m} }\myprob{X(a)=\alpha}=
  \left\{
    \begin{array}{ll}
      \mydel{a} & \text{, if $m=0$ and}\\
      \sum\limits_{q\in W} \myread{a}{q}\cdot\eta_q\cdot(1-\eta_q)^{m-1} & \text{, otherwise}
    \end{array}
  \right.
  \end{equation}
  and
  \begin{equation}\label{eq:23}
  \sum_{\alpha\in\myalple{m} }\myprob{\alpha\sqsubseteq X(a)}=
  \left\{
    \begin{array}{ll}
      1 & \text{, if $m=0$ and}\\
      \sum\limits_{q\in W} \myread{a}{q}\cdot(1-\eta_q)^{m-1} & \text{, otherwise.}
    \end{array}
  \right.
  \end{equation}
\end{proposition}

\begin{proof}
  Let $\myalp=\{a_1,\ldots,a_r\}$.
  We distinguish two cases: first, assume that $m=0$. Then $ \sum_{\alpha\in\myalple{m} }\myprob{X(a)=\alpha}= \mydel{a}$
  and $\myprob{\epsilon\sqsubseteq X(t)} = 1$, by definition. Second, assume that $m>0$. We first consider the
  term~\eqref{eq:22}: 
  \[
  \sum_{\alpha\in\myalple{m} }\myprob{X(a)=\alpha}=\sum_{b\in\myalp}\sum_{\beta\in\myalp^{m-1}}\sum_{q\in W}\myread{a}{q}\cdot\prod_{i=1}^{r}\mu_W(q,a_i,q)^{|\beta|_{a_i}}\mu_W(q,b,s)
  \]
 This can be seen as follows: in order to write a fixed string $\alpha=\beta b$ of length $m$, the PFA must move into some output
 state $q$ (with probability $\myread{a}{q}$) and then loop at $q$ exactly $m-1$ times, before it moves back into state
 $s$. While in state $q$, it must write the symbol $a_i$ exactly $|\beta|_{a_i}$ times for $i\in\{1,\ldots,r\}$. Also, it
 must write $b$ with returning into state $s$. The order in which the transitions must be made is given by $\beta$ 
 and thus the probability equals the product over the
 individual transition probabilities. Now, factoring out the respective terms according to the involved writing states
 and thereafter applying the multinomial theorem gives
 \begin{eqnarray*}
    \lefteqn{\sum_{b\in\myalp}\sum_{\beta\in\myalp^{m-1}}\sum_{q\in W}\myread{a}{q}\cdot\prod_{i=1}^{r}\mu_W(q,a_i,q)^{|\beta|_{a_i}}\mu_W(q,b,s)}\\\nonumber
    &=&\sum_{q\in W}\myread{a}{q}\cdot\sum_{b\in\myalp}\sum_{\beta\in\myalp^{m-1}}\prod_{i=1}^{r}\mu_W(q,a_i,q)^{|\beta|_{a_i}}\mu_W(q,b,s)\\\nonumber
   &=&\sum_{q\in W}\myread{a}{q}\cdot\Bigl(\sum_{b\in\myalp}\mu_W(q,b,s)\Bigr)\cdot\Bigl(\sum_{\beta\in\myalp^{m-1}}\prod_{i=1}^{r}\mu_W(q,a_i,q)^{|\beta|_{a_i}}\Bigr)\\\nonumber
   &=& \sum_{q\in W}\myread{a}{q}\cdot\eta_q\cdot(1-\eta_q)^{m-1}\mbox{}
 \end{eqnarray*} 
 Next, we consider the term~\eqref{eq:23}: We may proceed very similar, only that we must take into account the
 possibility that instead of returning from a writing state $q$ with the last symbol of $\alpha$, the PFA may again loop
 at $q$. This small difference manifests itself in the calculation as follows: 
 \begin{eqnarray*}
   \lefteqn{\sum_{\alpha\in\myalple{m} }\myprob{\alpha\sqsubseteq X(a)}}\\
    &=&\sum_{b\in\myalp}\sum_{\beta\in\myalp^{m-1}}\sum_{q\in W}\myread{a}{q}\cdot\prod_{i=1}^{r}\mu_W(q,a_i,q)^{|\beta|_{a_i}}(\mu_W(q,b,s)+\mu_W(q,b,q))\\\nonumber
    &=&\sum_{q\in W}\myread{a}{q}\cdot\sum_{b\in\myalp}\sum_{\beta\in\myalp^{m-1}}\prod_{i=1}^{r}\mu_W(q,a_i,q)^{|\beta|_{a_i}}(\mu_W(q,b,s)+\mu_W(q,b,q))\\\nonumber
   &=&\sum_{q\in W}\myread{a}{q}\cdot\Bigl(\sum_{b\in\myalp}\mu_W(q,b,s)+\mu_W(q,b,q)\Bigr)\cdot\Bigl(\sum_{\beta\in\myalp^{m-1}}\prod_{i=1}^{r}\mu_W(q,a_i,q)^{|\beta|_{a_i}}\Bigr)\\\nonumber
   &=& \sum_{q\in W}\myread{a}{q}\cdot(1-\eta_q)^{m-1}\mbox{}
 \end{eqnarray*}
 This proves the proposition.
 \end{proof}

Now, we have established the necessary notation to prove Lemma~\ref{lem:expanding-the-sum}.
\begin{proof}[Proof of Lemma~\ref{lem:expanding-the-sum}]
  Let $l\in\Nat_+$ and $t\in\myalpinf $.
  From Proposition~\ref{pro:decomposition1} we have
  \begin{eqnarray}\nonumber
    \sum_{\alpha\in\myalple{m} }\toalpx{X}{\alpha}{t}{1}{l}
    &=&\sum_{\alpha\in\myalple{m} }
    \sum_{\alpha=\beta_1\dots\beta_l} \prod^{l-1}_{i=1}
    \myprob{X(t[i])=\beta_i}\cdot\myprob{\beta_l\sqsubseteq X(t[l])}
    \\\nonumber
    &=&
    \sum_{m_1+\ldots+m_l=m} \prod^{l-1}_{i=1}
    \Bigl(\!\sum_{\beta_i\in\myalple{m_{i}} }\myprob{X(t[i])=\beta_i}\!\Bigr)\!\cdot\!
    \Bigl(\!\sum_{\beta_l\in\myalple{m_{l}} }\myprob{\beta_l\sqsubseteq X(t[l])}\!\Bigr)
  \end{eqnarray}
  Now,  for $i\in\{1,\ldots,l-1\}$ the term corresponding to $\beta_i$ 
  can according to Proposition~\ref{prop:decomp2} be can expanded as
\[
\sum_{\beta_i\in\myalple{m_{i}} }\myprob{X(t[i])=\beta_i}
=f(m_i)\!\cdot\!\mydel{t[i]}\!+\!
(1\!-\!f(m_i))\!\cdot\!\sum_{q\in W}\myread{t[i]}{q}\!\cdot\!\eta_q(1\!-\!\eta_q)^{m_i-1}\mbox{.}
\]
Also, the term corresponding to $\beta_l$ in the above product-sum can  according to Proposition~\ref{prop:decomp2} 
be expanded as
 \begin{eqnarray}\nonumber
\sum_{\beta_l\in\myalple{m_{l}} }\myprob{\beta_l\sqsubseteq X(t[l])}
&=&f(m_l)\!\cdot\!\mydel{t[l]}\!+\!
(1\!-\!f(m_l))\!\cdot\!\sum_{q\in W}\myread{t[l]}{q}\!\cdot\!(1\!-\!\eta_q)^{m_l-1}\\\nonumber
&\le&\frac{1}{\tilde{\eta}}\cdot\Bigl(f(m_l)\!\cdot\!\mydel{t[l]}\!+\!
(1\!-\!f(m_l))\!\cdot\!\sum_{q\in W}\myread{t[l]}{q}\!\cdot\!\eta_q(1\!-\!\eta_q)^{m_l-1}\Bigr)\mbox{.}
 \end{eqnarray}
Together, this gives
    \begin{eqnarray*}
      \sum_{\alpha\in\myalple{m} }\toalpx{X}{\alpha}{t}{1}{l}
      \le\frac{1}{\tilde{\eta}}\cdot\!\!\!\! 
      \sum_{m_1+\ldots+m_l=m}      
      \prod_{i=1}^{l}
      \biggl(
      f(m_i)\!\cdot\!\mydel{t[i]}\!+\!
      (1\!-\!f(m_i))\!\cdot\!\!\sum_{q\in W}\myread{t[i]}{q}\eta_q(1\!-\!\eta_q)^{m_i-1}
      \biggr)\mbox{,}
    \end{eqnarray*}
    which proves the lemma.
\end{proof}
}

\subsection{Bounding $\Phi(t,m,d)$}
\label{sec:bounding-phi}

In order to prove an exponentially decreasing upper bound on 
\[
\Phi(t,m,d)=2\cdot\!\!\sum_{\alpha\in\myalple{m} }\Bigl(\sum_{l=1}^{\lceil d\cdot m \rceil}\toalpx{X}{\alpha}{t}{1}{l}\Bigr)^2
\]
for fixed $d\in\Real_+$ an thereby prove Claim~\ref{claim:boundonpart1},
we first apply Cauchy's inequality and then bound by counting over all possible $l\in\Nat_+$:
\begin{eqnarray}\nonumber
2\cdot\!\!\sum_{\alpha\in\myalple{m} }\Bigl(\sum_{l=1}^{\lceil d\cdot m \rceil}\toalpx{X}{\alpha}{t}{1}{l}\Bigr)^2
&\le&2\lceil d\cdot m \rceil\cdot \!\sum_{\alpha\in\myalple{m} }\sum_{l=1}^{\lceil d\cdot m \rceil}\toalpx{X}{\alpha}{t}{1}{l}^2\\\label{eq:18}
&\le&2\lceil d\cdot m \rceil\cdot \!\sum_{l=1}^{\infty}\sum_{\alpha\in\myalple{m} }\toalpx{X}{\alpha}{t}{1}{l}^2
\end{eqnarray}
Here, exchanging the two sums does not change the value of the expression.
\paragraph{Bounding $\sum_{\alpha\in\myalple{m} }\toalpx{X}{\alpha}{t}{1}{l}^2$}
To proceed, we use the Conditions given by statement~$(1)$ of Theorem~\ref{thm:star-like} which as we will 
prove in Lemma~\ref{lem:corelemma} imply the existence of a
constant $\delta<1$ such that for all $l,m\in\Nat_+$ and $t\in\myalpinf $ 
the $l$-th addend of the outer sum of~\eqref{eq:18} can be bounded by 
$\sum_{\alpha\in\myalple{m} }\delta^{l-1}\cdot \toalpx{X}{\alpha}{t}{1}{l}$.
Before we proceed to the lemma, we state the following Proposition
which is a direct consequence of the definition of $\toalpx{X}{\alpha}{t}{k}{l}$.

\label{sec:prer-apply-tailb}
\begin{proposition}
  \label{pro:decomposition}
Let $j,k,l\in\Nat+$ satisfying $k\le j<l$ and $\alpha\in\myalple{m} $. Then
\begin{equation}\nonumber
  \toalpx{X}{\alpha}{t}{k}{l}=
  \sum_{i=0}^{m}\myprob{X(t[k\!\ldots\!j])=\alpha[0\!\ldots\!i]}\cdot\toalpx{X}{\alpha[i\!+\!1\ldots m]}{t}{j\!+\!1}{l}\mbox{.}
\end{equation}  
\end{proposition}

\begin{lemma}
  \label{lem:corelemma}
  Let  $X=(\{s\},W,\mu_R,\mu_W,\sigma)$ be a star-like PFA in canonical form over the finite alphabet $\myalp$. Let
    \begin{equation}
    \label{eq:1}
    \delta = \max_{a,b\in\myalp} \left(\mydel{a}+\sum_{q\in W}\myread{a}{q}\cdot(\mu_W(q,b,q)+\mu_W(q,b,s))\right)\mbox{.}
  \end{equation}  
  Then  for all infinite strings $t\in\myalpinf $ and all $k,l,m\in\Nat_+$ it holds that
  \begin{equation}\nonumber
    \label{implication:lem2}
    \sum_{\alpha\in\myalple{m} }\toalp{\alpha}{t}{k}{l}^2 \le \delta^{l-k}\!\cdot\!\sum_{\alpha\in\myalple{m} }\toalp{\alpha}{t}{k}{l}\mbox{.}    
  \end{equation}
\end{lemma}
\begin{proof}[Proof of Lemma~\ref{lem:corelemma}]
  Let $X$ be a star-like PFA and let $t\in\myalpinf $ be an arbitrary input string. 
  For $a\in\myalp$ and  $\alpha\in\myalpall $ satisfying $|\alpha|\ge 1$ 
  we have
  \begin{eqnarray}\nonumber
    \sum_{i=0}^{|\alpha|}\myprob{X(a)\!=\!\alpha[0\ldots i]} 
    &=& \myprob{X(a)=\varepsilon} + \myprob{X(a)=\alpha[1]} + \ldots \\\nonumber
    &\le & \myprob{X(a)=\varepsilon} + \myprob{\alpha[1]\sqsubseteq X(a)}\\\label{eq:lem2:3}
    &\le& \delta\mbox{.}
  \end{eqnarray}
  We prove the lemma by induction on the length $\ell=l-k+1$ of the part of $t$ which is read. 
  First note that for $l<k$  the left and the right hand side of Inequality~\eqref{implication:lem2}
  are equal to zero. This holds particularly, because $X$ is given in canonical form.
  Therefore we may without loss of generality assume that $l\ge k$ holds.\\
  {\bf  Induction basis:} for $\ell=1$ it holds that
  \[
  \sum_{\alpha\in\myalple{m} }\toalpxy{X}{\alpha}{t[l]}^2\le  \sum_{\alpha\in\myalple{m} }\toalpxy{X}{\alpha}{t[l]}\mbox{,}
  \]
  because probabilities are less than one.\\
  {\bf  Induction step:} assume that~\eqref{implication:lem2} holds for $l-k \le \ell-2$. By
  Proposition~\ref{pro:decomposition} we get
  \[
  \sum_{\alpha\in\myalple{m} }\toalpx{X}{\alpha}{t}{k}{l}^2=\sum_{\alpha\in\myalple{m} }\biggl(\sum_{i=0}^{m}\myprob{X(t[k])\!=\!
    \alpha[0\!\ldots\!i]}
  \cdot\toalpx{X}{\alpha[i\!+\!1\ldots\! m]}{t}{k+1}{l}\biggr)^2
  \]
  We apply Jensen's Inequality: let $x_0,\ldots,x_m,x_0,\ldots,x_m \in \Real^+$. Then
  \begin{equation}
    \label{eq:66b}
    \Bigl(\sum_{i=1}^{m}x_i y_i\Bigr)^2 = 
    \Bigl(\sum_{i=0}^{m}x_i\Bigr)^2\cdot\Bigl(\frac{\sum_{i=0}^{m}x_i y_i}{\sum_{i=0}^{m}x_i}\Bigr)^2
    \le \Bigl(\sum_{i=0}^{m}x_i\Bigr)\cdot\Bigl(\sum_{i=0}^{m}x_i y_i^2\Bigr)
  \end{equation}
  For $i\in\{0,\ldots,m\}$ we set
  \[
  x_i = \myprob{X(t[k])=\alpha[0\ldots i]}
  \]
  and
  \[
  y_i=\toalpx{X}{\alpha[i+1\ldots m]}{t}{k+1}{l}
  \]
  in Inequality~\eqref{eq:66b}. Additionally we know from Inequality~\eqref{eq:lem2:3} that
  \[
  \sum_{i=0}^{m}x_i = \sum_{i=0}^{m}\myprob{X(t[k])\!=\!\alpha[0\ldots i]} \le \delta\mbox{.}
  \]
  Together, we get that
  \[
\Bigl(\sum_{i=0}^{m}x_iy_i\Bigr)^2\le
  \Bigl(\sum_{i=0}^{m}x_i\Bigr)\cdot\Bigl(\sum_{i=0}^{m}x_i y_i^2\Bigr)
  \le \delta\cdot \Bigl(\sum_{i=0}^{m}x_i y_i^2\Bigr)\mbox{}
  \]
  which after re-translating gives
  \begin{eqnarray*}
    \lefteqn{\sum_{\alpha\in\myalple{m} }
      \biggl(\sum_{i=0}^{m}\myprob{X(t[k])\!=\! \alpha[0\!\ldots\!i]}
      \cdot\toalpx{X}{\alpha[i\!+\!1\!\ldots\!m]}{t}{k\!+\!1}{l}\biggr)^2}\\
    &&~~~~~~~~~~~\le    
    \sum_{\alpha\in\myalple{m} }
    \delta\cdot\sum_{i=0}^{m}
    \myprob{X(t[k])\!=\!\alpha[0\!\ldots\!i]}\cdot\toalpx{X}{\alpha[i\!+\!1\!\ldots \!m]}{t}{k\!+\!1}{l}^2\mbox{.}
  \end{eqnarray*}
  Using this we proceed as follows:
  \begin{eqnarray}\nonumber
    \lefteqn{   \sum_{\alpha\in\myalple{m} }\delta\cdot\sum_{i=0}^{m}
     \myprob{X(t[k])\!=\!\alpha[0\ldots i]}\toalpx{X}{\alpha[i+1\ldots m]}{t}{k+1}{l}^2}\\\nonumber
  &=&\delta\cdot\sum_{i=0}^{m}\Bigl(\sum_{\alpha_1\in\myalp^i}\myprob{X(t[k])\!=\!\alpha_1}\Bigr)\cdot
   \Bigl(\sum_{\alpha_2\in\myalp^{m-i}}\toalpx{X}{\alpha_2}{t}{i+1}{l}^2\Bigr)\\\label{eq:14}
   &\le&\delta\cdot\sum_{i=0}^{m}\Bigl(\sum_{\alpha_1\in\myalp^i}\myprob{X(t[k])\!=\!\alpha_1}\Bigr)\cdot
   \Bigl(\delta^{l-k-1}\sum_{\alpha_2\in\myalp^{m-i}}\toalpx{X}{\alpha_2}{t}{k+1}{j}\Bigr)\\\nonumber
   &=&\delta^{l-k}\cdot\sum_{i=0}^{m}\Bigl(\sum_{\alpha_1\in\myalp^i}\myprob{X(t[k])\!=\!\alpha_1}\Bigr)\cdot
   \Bigl(\sum_{\alpha_2\in\myalp^{m-i}}\toalpx{X}{\alpha_2}{t}{k+1}{l}\Bigr)\\\nonumber
   &=&
   \delta^{l-k}\cdot\sum_{\alpha\in\myalple{m} }\toalpx{X}{\alpha}{t}{k}{l}\mbox{.}
   \end{eqnarray}
   Here, Inequality~\eqref{eq:14} follows from the induction hypothesis. Altogether, we have shown
   \[
   \sum_{\alpha\in\myalple{m} }\toalpx{X}{\alpha}{t}{k}{l}^2\le
   \delta^{l-k}\cdot\sum_{\alpha\in\myalple{m} }\toalpx{X}{\alpha}{t}{k}{l}\mbox{.}
   \]
This proves the lemma.
\end{proof}

Lemma~\ref{lem:corelemma} tells us that the Conditions given by Statement~$(1)$ of Theorem~\ref{thm:star-like} allows us to bound
$\Phi(t,m,d)$ subject to the probability mass which is induced by the perturbation function $X$ on input
$t\in\myalpinf $ multiplied by a factor of $\delta<0$ for every input symbol which is read in the respective term. 
This is, using Inequality~\eqref{eq:18} from the beginning of this section and the lemma, we can bound   $\Phi(t,m,d)$ as
\[
  \Phi(t,m.d) 
  \le 2\lceil d\cdot m \rceil\cdot \!\sum_{l=1}^{\infty}\sum_{\alpha\in\myalple{m} }\toalpx{X}{\alpha}{t}{1}{l}^2\\\nonumber
  \le \frac{2\lceil d\cdot m \rceil}{\delta}\cdot \!\sum_{l=1}^{\infty}\sum_{\alpha\in\myalple{m} }\delta^{l}\cdot\toalpx{X}{\alpha}{t}{1}{l}\mbox{.}
\]
Now, we can expand (and bound) each addend of the last sum according to Lemma~\ref{lem:expanding-the-sum} as
    \begin{eqnarray}\nonumber
  \lefteqn{\sum_{\alpha\in\myalple{m} }\delta^{l}\cdot\toalpx{X}{\alpha}{t}{1}{l}}\\\label{eq:19}
&\le&
     1/\tilde{\eta}\cdot\!\!\!\!\!\!\!\!
   \sum_{m_1+\ldots+m_l=m}
      \prod_{i=1}^{l}
      \biggl(
      f(m_i)\!\cdot\!\delta \mydel{t[i]}\!+\!
      (1\!-\!f(m_i))\!\cdot\!\!\sum_{q\in W}\delta\myread{t[i]}{q} \eta_q(1\!-\!\eta_q)^{m_i-1}
      \biggr)\mbox{,}
    \end{eqnarray}
    where for $q\in W$, $\eta_q$ was defined as the return probability from state $q$. The minimum such probability was 
    $\tilde{\eta}=\min_{q\in W}\eta_q$ and $f:\Nat\rightarrow\{0,1\}$ was defined to be the indicator function of
    deleted letters. 
    
    \paragraph{Valid expressions}
    Call each non-zero addend in the above sum a {\em  valid expression} of $X$ on $t$. A valid expression is said to be
    of input length $l$ if it corresponds to a set of computations of input length $l$ and is said to be of output length $m$ if
    its corresponding set of computations has output length $m$.  Let $W=\{q_1,\ldots,q_{v}\}$.
    Each valid expression is a product over the set of variables 
    $\{\delta,\mydel{a_1},\mydel{a_r},\myread{a_1}{q_1},\ldots,\myread{a_r}{q_v},\eta_{q_1},\ldots,\eta_{q_1}\}$.
    The products have a regular structure which we exploit in order to bound the term~\eqref{eq:19}: to this end, let
    $ \myalpb$ be the following alphabet, where we interprete variables as letters (we intentionally use the term
    'letter' for an element of the alphabet and 'word' for a sequence of letters in order to avoid confusion)
    \begin{eqnarray*}
      \myalpb &=_\df&  \{\delta\}\\
      &\cup&\{\eta_q~:~q\in W\}\cup\{(1-\eta_q)~:~q\in W\}\\
      &\cup& \{\myread{a}{q}~:~q\in W\text{ and } a\in\myalp\} \cup \{\mydel{a}~:~a\in\myalp\}\mbox{.}
   \end{eqnarray*}
   Then, each valid expression in~\eqref{eq:19} is readily identifiable with a word $w$ over the alphabet 
    $\myalpb$: e.g., the word
    \[
    \delta\;\mydel{a}\;\delta\;\mydel{b}\;\delta\;\myread{a}{q}\;(1-\eta_q)\;(1-\eta_q)\;(1-\eta_q)\;\eta_q\mbox{}
    \]
    is corresponds to a valid expression of $X$ on $t=aba\ldots$ of input length $3$ and output length $4$ and thus
    to a set of computations of $X$ on $t$, where each computations deletes the first two letters and then moves into
    state $ q$ after having read the letter $a$, whereupon it loops three times at $q$ and then moves back to the input state again.
    Clearly, {\em  not all words} over $\myalpb$ are valid expression of $X$ on $t$. Call a word {\em  valid} if it does.
    Let $\mathcal{W}^{(t,l,m)}_X$ be  the set of all {\em  valid words} over $\myalpb$ that have input length
    $l$ and output length $m$, i.e., corresponding to a valid expression of $X$ on $t$ that has the respective input and
    output lengths. I.e., 
    \[
    \mathcal{W}_X^{(t,l,m)}=_\df\{w\in\myalpbfin~:~\text{$w$ is a valid word having input length $l$ and output length $m$}\}\mbox{.}
    \]    
    In order to evaluate the term~\eqref{eq:19} using the framework of valid expressions, 
    we follow the {\em  weighted words model}: we define the weight $\pi(w)$ of a word 
    $w\in\myalpbfin$ as the product of all letters which constitute $w$, where the multiplicity of a letter
    in the product equals the number of times it occurs in the word $w$:
    \[
    \pi(w)=_\df \sum_{x\in\myalpb} x^{|w|_x}\mbox{.}
    \] 
    E.g., for \[     w' = \delta\;\mydel{a}\;\delta\;\mydel{b}\;\delta\;\myread{a}{q}\;(1-\eta_q)\;(1-\eta_q)\;(1-\eta_q)\;\eta_q
    \]
    the example word from above, we have that
    \[
    \pi(w')=\delta^3\cdot\mydel{a}\cdot\mydel{b}\cdot\myread{a}{q}\cdot(1-\eta_q)^3\cdot\eta_q\mbox{,}
    \]
    as $|w'|_\delta=|w'|_{(1-\eta_q)}=3$ and
    $|w'|_{\mydel{a}}=|w'|_{\mydel{b}}=|w'|_{\myread{a}{q}}=|w'|_{\eta_q}=1$. 
    Also, the weight of a set is then defined as the weight of all elements in the set. Clearly, the weight of a valid
    word equals the value of its corresponding valid expression. It is easy to see that for $l,l'\in\Nat_+$ satisfying
    $l\neq l'$ it holds that $\mathcal{W}^{(t,l,m)}_X\cap\mathcal{W}^{(t,l',m)}_X=\emptyset$. Thus
    \begin{eqnarray*}
      \sum_{m_1+\ldots+m_l=m}
      \prod_{i=1}^{l}
      \biggl(
      f(m_i)\!\cdot\!\delta \mydel{t[i]}\!+\!
      (1\!-\!f(m_i))\!\cdot\!\!\sum_{q\in W}\delta\myread{t[i]}{q} \eta_q(1\!-\!\eta_q)^{m_i-1}
      \biggr)
      = \pi\left(\mathcal{W}_X^{(t,l,m)}\right)
    \end{eqnarray*}
    and therefore
    \begin{proposition}
      \label{prop:bound-by-valid-exp}
      \[
      \sum_{l=1}^{\infty}\sum_{\alpha\in\myalple{m} }\delta^{l}\cdot\toalpx{X}{\alpha}{t}{1}{l}
      \le
      \frac{1}{\tilde{\eta}\delta}\cdot\sum_{l=1}^{\infty} \pi\left(\mathcal{W}_X^{(t,l,m)}\right)\mbox{.}
      \]
    \end{proposition}
    Let $\mathcal{W}_X^{(t,l)}\supset \mathcal{W}_X^{(t,l,m)}$ be the set of all valid words having input length $l$,
    but {\em  arbitrary output length}. The set     $\mathcal{W}_X^{(t,l)}$ is a {\em  regular language}: 
    let '$|$' denote the choice operator and '$^*$' denote the sequence operator for
    a possibly zero number of repetitions of the respective letter. The 
    regular specification is as follows:
     \begin{eqnarray*}
         \mathcal{W}_X^{(t,l)}= \{&w&\in\myalpball\\
         &w&:=W_1\;W_2\;\dots\;W_l\\
         &W_1&:=\delta\;\mydel{t[1]}\;
         \;|\;\delta\;\myread{t[1]}{q_1}\;\eta_{q_1}\;(1-\eta_{q_1})^*\;|\;\ldots\;
         |\;\delta\;\myread{t[1]}{q_{v}}\;\eta_{q_{v}}\;(1-\eta_{q_{v}})^*\\
         &\vdots&\\
         &W_l&:=\delta\;\mydel{t[l]}\;
         \;|\;\delta\;\myread{t[l]}{q_1}\;\eta_{q_1}\;(1-\eta_{q_1})^*\;|\;\ldots\;
     |\;\delta\;\myread{t[l]}{q_{v}}\;\eta_{q_{v}}\;(1-\eta_{q_{v}})^*
     \}
   \end{eqnarray*}
   Here, we use $W_1,\ldots,W_l$ as {\em  placeholder} for the below defined regular expressions. 
   Clearly, for all words $ w\in\mathcal{W}_X^{(t,l)}$, it holds that $\sum_{i=1}^{l}|w|_{\delta} = l$.
   Now, we can formally define
    \[
    \mathcal{W}_X^{(t,l,m)}=_\df\{w\in \mathcal{W}_X^{(t,l)}~:~\sum_{i=1}^{v}(|w|_{\eta_{q_i}}+|w|_{(1-\eta_{q_i})}) =m\}\mbox{.}
    \]
 
     \paragraph{Embedding valid words}
     In order to evaluate the sum over all valid words of $X$ on $t$ of output length $m$, we first construct a
     family of structurally simpler languages such for each set $\mathcal{W}_X^{(t,l,m)}$ there exists a corresponding
     set in the structually simpler family having the same weight and such that the set in the family are still
     disjoint. Thus, the sum over the weights of all such new sets equals weight of valid words of $X$ on $t$ of
     output length $m$. Still, the sum over the weights of these new sets depends on the structure of the input string
     $t$ which is unknown. Thus, in order to get rid of this dependence on $t$, we do not evaluate the sum over the
     weights of all new sets exactly, but we over-count slightly. This over-counting can once again be best expressed by
     constructing a structually even simpler language which contains each word that we need to account for (ans some
     more words). To this end,  let  $\mathcal{Y}_X^{(t,l,m)}$ be
     the new set of words. For a fixed input string $t\in\myalpinf$ and a prefix $t[1\ldots l]$ let for $i\in\{1,\ldots, r\} $,
     \[
     l_i=|t[1\ldots l]|_{a_i}
     \] 
     be the number of occurrences of the symbol $a_i$ in the prefix $t[1\ldots l]$. For each such prefix we give a canonical input string
     $t'$ such that the corresponding set of valid words is structually simpler and 
     such that there is a function $g:\myalpbfin\rightarrow \myalpbfin$ that describes a bijection from
     $\mathcal{W}_X^{(t,l,m)}$ to $\mathcal{W}_X^{(t',l,m)}$, from which it follows that
     $\pi\left(\mathcal{W}_X^{(t,l,m)}\right)=\pi\left(\mathcal{W}_X^{(t',l,m)}\right)$.
     The string $t'$ is defined as
     \[
     t'=\underbrace{a_1\;\ldots\;a_1}_{l_i\text{ times}}\;\ldots\;\underbrace{a_r\;\ldots\;a_r}_{l_r\text{ times}}
     \]
     The corresponding sets $\mathcal{W}_X^{(t',l,m)}$ is then such that  for each word 
     $w\in\mathcal{W}_X^{(t,l,m)}$ there is a word $w'\in\mathcal{W}_X^{(t',l,m)}$ that is composed
     of exactly the same set of sub-words, but in different ordering: 
     Consider the decomposition of $w$ as
     \[
     w = W_1\;\ldots W_l\mbox{,}
     \]
     where the  $i$-th sub-word $W_i$ for $i\in\{1,\ldots,l\}$ corresponded to the symbol $t[i]$; Now, for $w'$ with the
     decomposition as 
     \[
     w' = W'_1\;\ldots W'_l
     \]
     it holds that the the first $l_1$ sub-words $W'_1,\ldots,W'_{l_1}$ correspond to the symbol $a_1$, 
     the next $l_2$ sub-words $W'_{l_1+1},\ldots,W'_{l_1+l_2}$ correspond to the symbol $a_2$, and so on. 
     Thus, the function $g$ is a permutation of sub-words: assume w.l.o.g. that the $i$-th subword $W_i$ 
     corresponds to the symbols $t[i]=a_j$. Then $W_i$ is mapped to the position
     $\sum_{k=1}^{j-1}l_k + |t[1\ldots i]|_{a_j}$ in $w'$. Such a permutation is clearly {\em  weight-preserving}.
     Now, define for a string $t$ the set $\mathcal{Y}_X^{(t,l,m)}$ as 
     \[
     \mathcal{Y}_X^{(t,l,m)} = _\df \mathcal{W}_X^{(t',l,m)}~~\text{where $t'$ is the canonical input string
       corresponding to $t[1\ldots l]$}\mbox{.}
     \]
     Clearly, for $l,l'\in\Nat_+$ is still holds that $     \mathcal{Y}_X^{(t,l,m)}\cap
     \mathcal{Y}_X^{(t,l',m)}=\emptyset$.
     \begin{proposition}
       \label{prop:lang-valid-ordered}
       \[
       \pi\left(\mathcal{W}_X^{(t,l,m)}\right)=\pi\left(\mathcal{Y}_X^{(t,l,m)}\right)
       \]
    \end{proposition}
    Now, we get rid of the dependency on the input string $t$: let $\mathcal{Y}_X$ be the following regular language over the alphabet 
    $\myalpb$, where again $Y_1,\ldots,Y_l$ are placeholder for regular expressions:
    \begin{eqnarray*}
      \mathcal{Y}_X= \{&w&\in\myalpball\\
      &w&:=Y_1\;Y_2\;\dots\;Y_r\\
      &Y_1&=\left(\delta\;\mydel{a_1}\;|\;
        \delta\;\myread{a_1}{q_1}\;\eta_{q_1}\;(1-\eta_{q_1})^*\;|\;\ldots\;|
        \delta\;\myread{a_1}{q_v}\;\eta_{q_v}\;(1-\eta_{q_v})^*
      \right)^*\\
      &\vdots&\\
      &Y_r&:=\left(\delta\;\mydel{a_r}\;|\;
        \delta\;\myread{a_r}{q_1}\;\eta_{q_1}\;(1-\eta_{q_1})^*\;|\;\ldots\;|
        \delta\;\myread{a_r}{q_v}\;\eta_{q_v}\;(1-\eta_{q_v})^*
      \right)^*     \}
    \end{eqnarray*}
    Clearly, $\mathcal{Y}^{(t,l,m)}\subset \mathcal{Y}_X$. We have, 
    \begin{eqnarray*}
      \mathcal{Y}_X^{(t,l,m)}&=& \{w \in\mathcal{Y}_X~:~\text{there exits $w'\in\mathcal{Y}^{(t,l,m)}$ such that $w=g(w')$}\}\\
      &=&   \{w\in\mathcal{Y}_X~:~\sum_{i=1}^{v}(|w|_{\eta_{q_i}}+|w|_{(1-\eta_{q_i})})=m\text{ and }\sum_{i=1}^{v}|w|_\delta=l\}\mbox{.}
    \end{eqnarray*}
    Now, define
    \[
    \mathcal{Y}^{(m)}_X=_\df 
    \{w\in\mathcal{Y}_X~:~\sum_{i=1}^{v}(|w|_{\eta_{q_i}}+|w|_{(1-\eta_{q_i})})=m\}\mbox{.}
    \]
   Clearly, for $m\in\Nat_+$ it holds that$     \mathcal{Y}_X^{(t,l,m)}\subset      \mathcal{Y}^{(m)}_X
   $
   and thus we have
   \begin{equation}
     \label{eq:15}
    \sum_{l=1}^{\infty}\pi\left(\mathcal{Y}_X^{(t,l,m)}\right) \le \pi\left(     \mathcal{Y}^{(m)}_X\right)\mbox{.}
   \end{equation}
   Altogether, we have established the following relation between the sum over all terms~\eqref{eq:19} over all
   $l\in\Nat_+$ and the weight of the set $\mathcal{Y}_x^{(m)}$:
   \begin{lemma}
     For $m\in\Nat_+$ and $t\in\myalpinf $, 
     \label{lem:final-words}\[
     \sum_{l=1}^{\infty}\sum_{\alpha\in\myalple{m} }\delta^{l}\cdot\toalpx{X}{\alpha}{t}{1}{l}
     \le
     1/(\tilde{\eta}\cdot\delta)\cdot\pi\left(\mathcal{Y}^{(m)}_X \right)\mbox{.}\]
     \end{lemma}
     \begin{proof}
       The lemma is easy to proof: 
       \begin{eqnarray*}
       \sum_{l=1}^{\infty}\sum_{\alpha\in\myalple{m} }\delta^{l}\cdot\toalpx{X}{\alpha}{t}{1}{l}
       &=&1/(\tilde{\eta}\cdot\delta)\cdot\sum_{l=1}^{\infty}\pi\left(\mathcal{W}_X^{(t,l,m)}\right)
       ~~~~~~\text{by Proposition~\ref{prop:bound-by-valid-exp}}\\
       &=&1/(\tilde{\eta}\cdot\delta)\cdot\sum_{l=1}^{\infty}\pi\left(\mathcal{Y}_X^{(t,l,m)}\right)
       ~~~~~~~\text{by Proposition~\ref{prop:lang-valid-ordered}}\\
       &\le&1/(\tilde{\eta}\cdot\delta)\cdot
       \pi\left(\mathcal{Y}^{(m)}_X \right)~~~~~~~~~~~~~~~~~~~~~~~~~~~\text{by~\eqref{eq:15}}
     \end{eqnarray*}
       This proves the lemma.
     \end{proof}
     
     \paragraph{A crude bound on $\pi\left(\mathcal{Y}^{(m)}_X \right)$ using the saddle point method}
     In order to get a bound on the term $\pi\left(\mathcal{Y}^{(m)}_X \right)$,
     we proceed as follows: after having given the regular specification of the set $\mathcal{Y}_X$, we
     translate this specification into the language of generating
     functions, where we use the variable $z$ to mark the length, i.e., 
     the number $\sum_{i=1}^{v}(|w|_{\eta_{q_i}}+|w|_{(1-\eta_{q_i})})$ for a word $w\in\mathcal{Y}_X$. 
     Also, we symbolically use the letters as variables.
     A regular specification for a set of combinatorial objects translates into a {\em  invariably positive rational}
     generating function, where we have the following relationship between the operators of the regular description
     and the algebraic operators: let $a,b\in\myalpb$. Then {\em  union}, i.e., '$a|b$', corresponds to
     '${a}+{b}$', {\em  combinatorial product}
     , i.e., '$a\; b$', corresponds to '${a} \cdot {b}$' and {\em  sequence building}, i.e., '$a^*$', corresponds to
     $1/(1-{a})$ (where $a\neq \epsilon$). Thus, the regular specification of the language $\mathcal{Y}_X$ readily lends
     itself to the following  ordinary multivariate generating function. 
     
     \begin{lemma}\label{lem:generatingf}
     The ordinary multivariate generating function corresponding to the language $\mathcal{Y}_X$
     is
     \[
     \mathcal{Z}_X(\mathbf{z},{\delta},\vec{{\eta}},\vec{{\rho}})
     =\prod_{i=1}^{r}\left(
         1-
         \delta\!\cdot\!\mydel{a_i} - 
         \sum\limits_{j=1}^{v} \frac{\delta\!\cdot\!\myread{a_i}{q_j}\!\cdot\!\eta_{q_j}\!\cdot\!\mathbf{z}}{1-(1-\eta_{q_j})\!\cdot\!\mathbf{z}}
       \right)^{-1}
     \]
     where $\vec{{\eta}}=({\eta}_{q_1},\ldots,{\eta}_{q_k},{(1-\eta_{q_1})},\ldots,{(1-\eta_{q_v})})$ and
     $\vec{\rho}=({\mydel{t[1]}},\ldots,{\mydel{t[i]}},{\myread{t[1]}{q_1}},\ldots{\myread{t[l]}{q_v}})$. Here the
     variable $z$ marks the number $\sum_{i=1}^{v}(|w|_{\eta_{q_i}}+|w|_{(1-\eta_{q_i})})$ and the other variables mark
     the number of occurences of the respective letters.
   \end{lemma}
\begin{proof}
     In order to make the proof more readable, we mark the number of occurences of a letter by a variable with the
     corresponding latin symbol.
     \begin{itemize}
     \item $\delta$ is marked by $d$.
     \item For $i\in\{1,\ldots,r\}$, $\mydel{a_i}$ is marked by $r_{a_i}$.
     \item For $i\in\{1,\ldots,r\}$ and $j\in\{1,\ldots,v\}$, $\myread{a_i}{q_j}$ is marked by $r_{a_i,q_j}$.
     \item For $j\in\{1,\ldots,v\}$, $\eta_{q_j}$ is marked by $e_{q_j}$ and
     \item For $j\in\{1,\ldots,v\}$, $(1-\eta_{q_j})$ is marked by $(1-e_{q_j})$.
     \end{itemize}
     Also, $z$ marks the number $\sum_{i=1}^{v}(|w|_{\eta_{q_i}}+|w|_{(1-\eta_{q_i})})$.
     Consider the $i$-th addend in the product for $i\in\{1,\ldots r\}$. The set of words $\{\delta\;\mydel{a_i}\}$ is
     generated by the mulitvariate generating function (MGF)
     \[
     g_i(d,r_{a_i}) = d\cdot r_{a_i}\mbox{.}
     \]
     For $j\in\{1,\ldots,v\}$, the set of words 
     \[
     \{\delta\;\myread{a_i}{q_j}\;\eta_{q_j}, \delta\;\myread{a_i}{q_j}\;\eta_{q_j}(1\!-\!\eta_{q_j}),\delta\;\myread{a_i}{q_j}\;\eta_{q_j}(1\!-\!\eta_{q_j})(1\!-\!\eta_{q_j}),\ldots\}
     \]
     is generated by the MGF
     \[
     f_{ij}(z,d,r_{a_i,q_j},e_{q_j},(1\!-\!e_{q_j})) = \frac{d\cdot r_{a_i,q_j}\cdot e_{q_j}\cdot z}{1\!-\!(1\!-\!e_{q_j})\cdot z}
     \]
     Now the words corresponding to state $q_j$ are generated by the regular expression
     \begin{equation}
       \label{eq:10}
     \left(\delta\;\mydel{a_{i}}\;|\;
       \delta\;\myread{a_{i}}{q_{i}}\;\eta_{q_{i}}\;(1\!-\!\eta_{q_{i}})^*\;|\;\ldots\;|
       \delta\;\myread{a_{i}}{q_v}\;\eta_{q_v}\;(1\!-\!\eta_{q_v})^*
     \right)^*     
    \end{equation}
    are generated by the function
     \begin{eqnarray*}
\lefteqn{     f_i(z,d,r_{a_i},r_{a_i,q_1},\ldots,r_{a_i,q_v},e_{q_1},\ldots,e_{q_1},(1\!-\!e_{q_1}),\dots,(1\!-\!e_{q_1}))}\\
     &&~~~~~~~~~~~~~~~~~=
     \left(1 - g_i(d,r_{a_i}) -\sum_{j=1}^{v}     f_{ij}(z,d,r_{a_i,q_j},e_{q_j},(1\!-\!e_{q_j}))\right)^{-1}\\
     &&~~~~~~~~~~~~~~~~~=
     \left(1 - d\cdot r_{a_i} -\sum_{j=1}^{v}\frac{d\cdot r_{a_i,q_j}\cdot e_{q_j}\cdot z}{1\!-\!(1\!-\!e_{q_j})\cdot z}\right)^{-1}\mbox{.}
     \end{eqnarray*}
     Now, the set $\mathcal{Y}_X$, which is defined by a regular expression that is the concatenation of the regular
     expression~\eqref{eq:10} for state $q_j$  for $j\in\{1,\ldots,v\}$ is generated the the product over the
     corresponding MGF's. Resubstituting the respective variables proves the Lemma.
   \end{proof}

   In order to evaluate  $\pi\left(\mathcal{Y}^{(m)}_X \right)$, we follow the {\em  weighted words model}: this is, the former variables
   are treated as parameters.in the new generating function. The respective function is then
   \[   
     \tilde{\mathcal{Z}}_X(z)
     =\prod_{i=1}^{r}\left(
         1-
         \delta\!\cdot\!\mydel{a_i} - 
         \sum\limits_{j=1}^{v} \frac{\delta\!\cdot\!\myread{a_i}{q_j}\!\cdot\!\eta_{q_j}\!\cdot\! z}{1-(1-\eta_{q_j})\!\cdot\! z}
       \right)^{-1}\mbox{.}
   \]
   Now, there are (at least) two ways to proceed in order to derive the weight 
   $\pi\left(\mathcal{Y}^{(m)}_X \right)$: since the function $\tilde{\mathcal{Z}}_X(z)$ is a rational function,
   it lends itself to a partial fraction decomposition. Then, one can easily translate this form back into
   a formal power series  $A(z)=\sum_{i=1}^{\infty}a_iz^{i}$ and $\pi\left(\mathcal{Y}^{(m)}_X \right)$ equals the coefficient at $z^m$ of this power
   series, i.e.,
   \[
   \pi\left(\mathcal{Y}^{(m)}_X \right) = \left[z^m\right]A(z) = a_m\mbox{.}
   \]
   Since a partial fraction decomposition of the function $\tilde{\mathcal{Z}}_X(z)$ is quite involved, we do
   not follow this vein here: instead, we use the following Theorem on the expansion of rational functions
   
   \begin{theorem}[Expansion of rational functions]{{\em [Theorem~IV in~\cite{Flajolet-2007}]}}
     \label{thm:rational-expansion}
     If $f(z)$ is a rational functions that is analytic at zero and has poles at $z_1\le z_2\le\ldots\le z_k$ then
     its coefficients are a sum of  {\em  exponential polynomials}: there exist $k$ polynomials $\Pi_1(z),\ldots,\Pi_k(z)$
     such that for $m$ larger than some fixed $m_0$,
     \[
     \left[z^m\right]f(z)=\sum_{j=1}^{k}\Pi_j(m)\cdot \left(\frac{1}{z_j}\right)^m\mbox{.}
     \]
     Furthermore, the polynomial $Pi_j$ has degree equal to the order of the pole at $z_j$ minus one.
   \end{theorem}
   
   By construction of the regular language $\mathcal{Y}_X$, all poles of $\tilde{\mathcal{Z}}$ are of order at most $r$, where $r$ is the
   cardinality of $\myalp$. Let $\tilde{z}_1,\ldots,\tilde{z}_{r'}$ where $r'\in\{1,\ldots,r\}$ be these poles (which have not yet been
   specified) and let $\tilde{z}_1$ the pole of smallest modulus. Then according to the above theorem we have that
   \begin{equation}
     \label{eq:1a}
   \left[z^m\right] \tilde{\mathcal{Z}}_X = \pi\left(\mathcal{Y}^{(m)}_X\right)
   = \sum_{i=1}^{r'} \Pi_i(m)\cdot\left( \frac{1}{\tilde{z}_i}\right)^m \le
   \left( \frac{1}{\tilde{z}_1}\right)^m  \cdot \sum_{i=1}^{r'} \Pi_i(m)\mbox{,}
 \end{equation}
 where for $i\in\{1,\ldots,r'\}$ is a polynomial of degree at most equal to the order of the pole at $z_i$ minus one.
  Now, we are in a position to prove the exponentially decreasing upper bound on $\Phi(t,m,d)$.
  
 \begin{proof}[Proof of Claim~\ref{claim:boundonpart1}]
   Recapitulating the previous calculation, we have
   \begin{eqnarray}\nonumber
     \Phi(t,m,d) &=& 2\cdot\!\!
     \sum_{\alpha\in\myalple{m} }\Bigl(\sum_{l=1}^{\lceil d\cdot m \rceil}\toalpx{X}{\alpha}{t}{1}{l}\Bigr)^2
     ~~~~~~~~~~~~~~~~~~~~~~~~~~~~~~~~~\text{(Def.)}\\\nonumber
     &\le&2\lceil d\cdot m \rceil\cdot \!\sum_{l=1}^{\infty}
     \sum_{\alpha\in\myalple{m} }\toalpx{X}{\alpha}{t}{1}{l}^2
     ~~~~~~~~~~~~~~~~~~~~~\text{(Ineq.~\eqref{eq:18})}\\\nonumber
     &\le& \frac{2\lceil d\cdot m \rceil}{\delta}\cdot \!\sum_{l=1}^{\infty}
     \sum_{\alpha\in\myalple{m} }\delta^{l}\cdot\toalpx{X}{\alpha}{t}{1}{l}
  ~~~~~~~~~~~~~~~~~~~\text{(Lem.~\ref{lem:corelemma})}\\\nonumber
  &\le &       \frac{2\lceil d\cdot m \rceil}{\tilde{\eta}\delta}\cdot\sum_{l=1}^{\infty} \pi\left(\mathcal{Y}_X^{(m)}\right)
  ~~~~~~~~~~~~~~~~~~~~~~~~~~~~~~~~~~~~\text{(Lem.~\ref{lem:final-words})}\\\nonumber
  &\le & \frac{2\lceil d\cdot m \rceil}{\tilde{\eta}\delta}\cdot \left[z^m\right]      \tilde{\mathcal{Z}}_X(z)
  ~~~~~~~~~~~~~~~~~~~~~~~~~~~~~~~~~~~~~~~~\text{(Lem.~\ref{lem:generatingf})}\\\nonumber
  &\le &  \frac{2\lceil d\cdot m \rceil}{\tilde{\eta}\delta}\cdot
    \sum_{i=1}^{r'} \Pi_i(m)\cdot \left(\frac{1}{\tilde{z}_1}\right)^m
  ~~~~~~~~~~\text{
    (Eq.~\eqref{eq:1a} and Thm.~\ref{thm:rational-expansion})}\mbox{,}
\end{eqnarray}
where $\tilde{z}_1$ is the pole of minimum modulus of the function $\tilde{\mathcal{Z}}_X(z)$. Now, since $\lceil d\cdot m\rceil<
(d+1)\cdot m$ the Claim follows with $\Pi(m)=\frac{2(d+1)}{\tilde{\eta}\delta}\cdot \sum_{i=1}^{r'} \Pi_i(m)$.
\end{proof}

 \subsection{Bounding $\Psi(t,m,d)$}\label{sec:bounding-psi}
 In this section, we derive the exponentially decreasing upper bound on the term $\Psi(t,m,d)$ for fixed $d\in\Real_+$.
 Remember that we fixed
 \[
 d =_\df \min\{d' \in \Real~:~ (\sqrt[d']{e}\cdot p_\mathrm{max}<1)~\wedge~
 (d'\cdot (p_\mathrm{max})^{d'}<\frac{(p_\mathrm{max})^2}{3})\}\mbox{,}
 \]
where $p_\mathrm{max} = \max_{a\in\myalp}\mydel{a}$ was the maximum deletion probability. Set 
 \[
 \gamma_{\text{\ref{claim:boundonpart2}}}=_\df e\cdot(d+1)\cdot(p_\mathrm{max})^{d-1}
 \]
 and
 \[
 c_{\text{\ref{claim:boundonpart2}}}=_\df \frac{\sqrt[d+1]{e}}{\tilde{\eta}\cdot(1-\sqrt[d+1]{e}\cdot p_\mathrm{max})}\mbox{.}
 \]
 Here,  $e\approx 2.71...$ is the base of the natural logarithm and $\tilde{\eta}=\min_{q\in W}\eta_q$.
We prove Claim~\ref{claim:boundonpart2} by showing that for the above choice of constants it holds that
\[
\Psi(t,m,d) \le c_{\text{\ref{claim:boundonpart2}}}\cdot (\gamma_{\text{\ref{claim:boundonpart2}}})^m\mbox{.}
\]
The choice of $d$ gives that $ \gamma_{\text{\ref{claim:boundonpart2}}}<1$. This justifies the choice.
\begin{proof}[Proof of Claim~\ref{claim:boundonpart2}]
  Let $d$ and $p_\mathrm{max}$ be defined as above. We start as follows:
  \begin{eqnarray*}
    \Psi(t,m,d)&=&2\cdot\sum_{\alpha\in\myalple{m} }\Bigl(\sum_{l=\lceil d \cdot m \rceil+1}^{\infty}\toalp{\alpha}{t}{1}{l}
    \Bigr)^2\\
          &\le&
          \sum_{\alpha\in\myalple{m} }\sum_{l=\lceil d \cdot m \rceil+1}^{\infty}\toalp{\alpha}{t}{1}{l}\\
          &=&    \sum_{l=\lceil d \cdot m \rceil+1}^{\infty}    \sum_{\alpha\in\myalple{m} }\toalp{\alpha}{t}{1}{l}\mbox{.}  
        \end{eqnarray*}
        This holds particularly, because we deal with probabilities, i.e., quantities less than one. Thus, we have
        bounded $\Psi(t,m,d)$ by the that part of the probability mass induced by $X$ on input $t$ which corresponds to
        the cases in which $X$ has read a relatively long prefix of $t$.
        Next, we consider the expansion of $\sum_{\alpha\in\myalple{m} }\toalp{\alpha}{t}{1}{l}$ for a fixed $l\ge \lceil d
       \cdot m \rceil+1$ due to Lemma~\ref{lem:expanding-the-sum}:
       \begin{eqnarray}\nonumber
        \lefteqn{\sum_{\alpha\in\myalple{m} }\toalp{\alpha}{t}{1}{l}}\\\nonumber
        &\le&\frac{1}{\tilde{\eta}}\cdot\!\!\!\! 
      \sum_{m_1+\ldots+m_l=m}      
      \prod_{i=1}^{l}
      \biggl(
      f(m_i)\!\cdot\!\mydel{t[i]}\!+\!
      (1\!-\!f(m_i))\!\cdot\!\!\sum_{q\in W}\myread{t[i]}{q}\eta_q(1\!-\!\eta_q)^{m_i-1}
      \biggr)\\\label{eq:17}
      &\le&
      \frac{1}{\tilde{\eta}}\cdot\binom{m+l-1}{l-1}\cdot (p_\mathrm{max})^{l-m}
      \mbox{,}
\end{eqnarray}
    Inequality~\eqref{eq:17} follows
    from the fact that for $l\ge \lceil d \cdot m \rceil +1$ in every decomposition
    $m=m_1+\ldots +m_l$ of $m$ into $l$ non-negative addends, there are at least $l-m$ indices $i$, where 
    $1\le i\le l$ such that $m_i=0$. For each such $m_i$, it holds that $f(m_i)=1$ and thus a factor of
    $\mydel{t[i]}\le p_\mathrm{max}$ is ``added'' in the product. Also, there are at most $\binom{m+l-1}{l-1}$
    such decompositions. Using Stirling's Approximation for the Binomial Coefficient and the Fact that $(n+1/n)^n<e$ we 
    may further bound as follows:
\begin{eqnarray}\nonumber
 \sum_{l=\lceil d \cdot m \rceil+1}^{\infty}    \sum_{\alpha\in\myalple{m} }\toalp{\alpha}{t}{1}{l}
      &\le&
      \frac{1}{\tilde{\eta}}\cdot\sum_{l=\lceil d\cdot m\rceil +1}^{\infty}
      \binom{m+l-1}{ l-1}\cdot(p_\mathrm{max})^{l-m}\\\nonumber
      &\le&\frac{1}{\tilde{\eta}}\cdot(p_\mathrm{max})^{(d-1)m}\cdot\!
      \sum_{l=0}^{\infty}
      \binom{\lceil d\cdot m\rceil + m +l}{m}\cdot(p_\mathrm{max})^{l}\\\nonumber
      &\le&\frac{\sqrt[d+1]{e}}{\tilde{\eta}}\cdot\bigl(e(d+1)(p_\mathrm{max})^{(d-1)}\bigr)^{m}\cdot\!
      \sum_{l=0}^{\infty} (\sqrt[d+1]{e}\cdot p_\mathrm{max})^l\\\label{eq:12}
      &=&
      \frac{
        \sqrt[d+1]{e}
      }{
      \tilde{\eta}\cdot(1-\sqrt[d+1]{e}\cdot p_\mathrm{max})
      }\cdot \bigl(e(d+1)(p_\mathrm{max})^{(d-1)}\bigr)^{m}\\\nonumber
      &=&c_{\text{\ref{claim:boundonpart2}}}\cdot(\gamma_{\text{\ref{claim:boundonpart2}}})^m\mbox{.}
\end{eqnarray}
Here,~\eqref{eq:12} holds, because  $\sqrt[d+1]{e}\cdot p_\mathrm{max}<1$ by our choice of $d$.
Hence, Claim~\ref{claim:boundonpart2} follows.
    \end{proof}
    
    \begin{remark}
      \label{remark1}
      Note that $\gamma_{\text{\ref{claim:boundonpart2}}}$ can be made {\em  arbitrarily} small, as
      $\lim_{d\rightarrow \infty}d\cdot (p_\mathrm{max})^d = 0$. Our choice of $d$ being minimal such that 
      the exponentially decreasing upper bound on $\Psi(t,m,d)$  can be shown can thus be improved such that 
      for $d$ sufficiently large,
      \[
      \lim_{m\rightarrow\infty }\frac{\Phi(t,m,d)}{\Psi(t,m,d)}=0
      \]
      and therefore the base of the logarithm for the smoothed trie height depends {\em  only} on the upper bound on 
      $\Phi(t,m,d)$.
    \end{remark}

    In smoothed analysis it is usual to quantify the influence of the perturbation function on the smoothed complexity.
    Here, the respective quality is the trie height. So far, we have ignored the quantitative influence of the
    perturbation function and have only given a qualitative result.  Note that by Remark~\ref{remark1} immediatly
    implies Theorem~\ref{thm:tight-for-star-like} .

\end{document}